\documentclass[]{article}

\makeatletter
\g@addto@macro\bfseries{\boldmath}
\makeatother

\usepackage{graphicx}
\usepackage{amsmath,amscd}
\usepackage{amssymb}
\usepackage{mathrsfs}
\usepackage[colorlinks]{hyperref}
\usepackage{breakurl}
\usepackage{mathabx,epsfig}
\usepackage{epstopdf}
\usepackage[cmtip,all]{xy}

\usepackage{amsthm}

\usepackage{tikz}
\usetikzlibrary{decorations.markings}

\hypersetup{
allcolors={black}
}

\usepackage[margin=1.795in]{geometry}

\usepackage[normalem]{ulem}

\newcommand{\lmatrp}[1]{\ensuremath{\left(\begin{array}{#1}}}
\newcommand{\rmatrp}{\ensuremath{\end{array}\right)}}
\newcommand{\df}{\ensuremath{\mathrel{\mathop:}=}}

\newcommand{\im}{\mathrm{im}\,}
\newcommand{\id}{\mathrm{id}}

\newtheorem{theorem}{Theorem} [section]
\newtheorem{proposition}[theorem]{Proposition}
\newtheorem{corollary}[theorem]{Corollary}
\newtheorem{lemma}[theorem]{Lemma}
\theoremstyle{definition}

\theoremstyle{remark}
\newtheorem{remark}[theorem]{Remark}
\newtheorem{example}[theorem]{Example}

\allowdisplaybreaks

\begin{document}
    \title{$K$-theory of AF-algebras from\\ braided C*-tensor categories}
    \author{Andreas N\ae s Aaserud and David Emrys Evans}

    \maketitle

    \begin{abstract}
\noindent Renault, Wassermann, Handelman and Rossmann (early 1980s) and Evans and Gould (1994) explicitly described the $K$-theory of certain unital AF-algebras $A$ as (quotients of) polynomial rings.
        In this paper, we show that in each case the multiplication in the 
        polynomial ring (quotient) is induced by a $*$-homomorphism $A\otimes 
        A\to A$ arising from a unitary braiding on a C*-tensor category and 
        essentially defined by Erlijman and Wenzl (2007). We also present some 
        new explicit calculations based on the work of Gepner, Fuchs and 
        others. Specifically, we perform computations for the rank two compact 
        Lie groups SU(3), Sp(4) and G$_2$ that are analogous to the 
        Evans--Gould computation for the rank one compact Lie group SU(2).

        The Verlinde rings are the fusion rings of Wess--Zumino--Witten models 
        in conformal field theory or, equivalently, of certain related 
        C*-tensor categories.
        Freed, Hopkins and Teleman (early 2000s) realized these rings via twisted equivariant $K$-theory. Inspired by this, our long-term goal is to realize these rings in a simpler $K$-theoretical manner, avoiding the technicalities of loop group analysis. As a step in this direction, we note that the Verlinde rings can be recovered as above in certain special cases.
    \end{abstract}

    \tableofcontents

    \section{Introduction}

    This paper concerns the $K$-theory of certain unital AF-algebras. Recall that AF-algebras are limits of inductive sequences of finite-dimensional C*-algebras with connecting $*$-homomorphisms. They were first introduced by Bratteli in the paper \cite{Br}. The $K$-theory $K_0(A)$ of an AF-algebra $A$ is an ordered group, i.e., an abelian (additively written) group $G$ together with a positive cone $G_+$, i.e., a subset $G_+\subset G$ that satisfies $G_++G_+\subset G_+$, $G_+\cap (-G_+) = \{0\}$ and $G = G_+-G_+$. If $A$ is a unital AF-algebra then $K_0(A)$ is an ordered group with a distinguished order unit (cf.\ Remark \ref{remark:SU3}), and Elliott \cite{Ell} showed that unital AF-algebras are classified by this data. He also classified non-unital AF-algebras by a slightly more complicated invariant. The ordered groups that arise as $K_0(A)$ for some AF-algebra $A$ are called dimension groups. They were abstractly characterized by Effros, Handelman and Shen in \cite{EHS}. We refer the reader to Effros' monograph \cite{E} for more information on the $K$-theory of AF-algebras.

    Although $K_0(A)$ is in general only an ordered group,
    it also sometimes has
    a natural ordered ring structure (cf.\ section \ref{subsection:ring_structure}). By the Elliott Classification Theorem \cite{Ell}, if $A$ is a unital AF-algebra such that $K_0(A)$ has the structure of an ordered ring then the product on $K_0(A)$ is induced by a (non-explicit) $*$-homomorphism $A\otimes A\to A$. In the present paper, we consider certain unital AF-algebras $A$ for which $K_0(A)$ has the structure of an ordered ring and the underlying $*$-homomorphism $A\otimes A\to A$ can be defined in terms of a unitary braiding on a rigid C*-tensor category. (These concepts are described in some detail in section \ref{subsection:categories} below, and some historical background and references can be found in Remark \ref{remark:history}.)
    We place certain $K$-theory calculations
    of
    Renault, Wassermann, Handelman and Rossmann (early 1980s) and Evans--Gould 
    (1994) within this framework and, in some special cases, give new explicit 
    descriptions of the ring $K_0(A)$ in terms of generators and relations, 
    i.e., as a quotient of a polynomial ring by an ideal generated by certain 
    polynomials. Below, we motivate our computations and put them into a 
    historical context.

    In the early 2000s, Freed, Hopkins and Teleman used $K$-theory to study 
    Wess--Zumino--Witten (WZW) models in 2d conformal field theory. Each WZW 
    model has an underlying rigid C*-tensor category (in fact, a unitary 
    modular tensor category) $\mathbf{Rep}_k(G)$ with a unitary braiding. Here, 
    $G$ is a simple, connected, simply connected, compact Lie group and $k$ is 
    a positive integer known as the ``level'' (cf.\ Example \ref{example:RepkG} 
    below). The fusion ring of this category, which is known as the Verlinde 
    ring $\mathrm{Ver}_k(G)$, describes the so-called operator product 
    expansion of primary fields in the model. Freed, Hopkins and Teleman 
    described $\mathrm{Ver}_k(G)$ in terms of twisted equivariant $K$-theory 
    (cf.\ \cite{FHT1,FHT2,FHT3}). In symbols, those authors proved that, as 
    rings,
    $$
        \mathrm{Ver}_k(G) \cong {}^\tau K_{G}^{\dim(G)}(G),
    $$
    where $G$ acts on itself by conjugation, the twist $\tau = \tau(k)$ belongs to the equivariant \v{C}ech cohomology group $H^3_{G}(G;\mathbb{Z})\cong \mathbb{Z}$, and the ring structure on the right hand side arises from the group product. (The above formula is actually a special case of a more general theorem from \cite{FHT1,FHT2,FHT3}.) Note that ${}^\tau K_{G}^{\dim(G)}(G)$ can also be viewed, in terms of the graded $K$-theory of C*-algebras, as $K_{\dim(G)}^{G}(C_0(G,\mathbb{K}_\tau))$ or, equivalently, as $K_{\dim(G)}(C_0(G,\mathbb{K}_\tau)\rtimes G)$, where $C_0(G,\mathbb{K}_\tau)$ is the C*-algebra of sections of a certain twisted $G$-equivariant graded bundle $\mathbb{K}_\tau$ of compact operators over $G$.

    Inspired by this, a long-term goal of ours is to find a simpler $K$-theoretical description of these rings, avoiding the technically complicated analysis of loop groups used in \cite{FHT1,FHT2,FHT3} and perhaps even allowing us to describe the underlying rigid C*-tensor category in terms of modules over a ``natural'' (possibly non-AF) C*-algebra along with a ``natural'' tensor product on these modules. This served as one impetus for
    the present paper. Since it was written, we have actually accomplished something along these lines. Namely, in \cite{AE} (see Remark 5.6 therein) we use the graphical calculus for $\mathbf{Rep}_k(G)$ (see Remark \ref{remark:history} below) to realize this category as a braided C*-tensor category of Hilbert C*-modules over a C*-algebra of compact operators. In the present paper, we content ourselves with realizing --- in certain special cases --- the ring $\mathrm{Ver}_k(G)$, equipped with a natural positive cone, as the ordered $K_0$-ring of a unital AF-algebra arising from the category $\mathbf{Rep}_k(G)$.

    $K$-theory has also been used to study certain actions of compact groups on 
    AF-algebras, for instance in the work of Wassermann \cite{Wa}, 
    Handelman--Rossmann \cite{HR1,HR2} and Handelman \cite{H1,H2} in the 1980s. 
    As part of their work, those authors explicitly described the $K$-theory of 
    the corresponding fixed point AF-algebras as polynomial rings. For example, 
    inspired by Renault's computation of $K_0(M_{2^\infty}^{\mathbb{T}})$ (cf.\ 
    the appendix to \cite{Re}), Wassermann showed that 
    $K_0(M_{2^\infty}^{\mathrm{SU}(2)})\cong \mathbb{Z}[t]$ as ordered rings 
    (cf.\ \cite{Wa}, pp.\ 118--123; see also Theorem \ref{theorem:SU2} below), 
    while Handelman and Rossmann showed that 
    $K_0(M_{n^\infty}^{\mathrm{U}(n)})\cong \mathbb{Z}[t_1,\ldots,t_{n-1}]$ as 
    rings (cf.\ Proposition VII.1 in \cite{HR1}). In these formulae, 
    $M_{n^\infty}$ is the infinite tensor product $M_n(\mathbb{C})\otimes 
    M_n(\mathbb{C})\otimes\cdots$ and the action of $\mathrm{SU}(2)$ on 
    $M_{2^\infty}$ (resp.\ of $\mathrm{U}(n)$ on $M_{n^\infty}$) arises from 
    the action of $\mathrm{SU}(2)$ on $M_2(\mathbb{C})$ (resp.\ of 
    $\mathrm{U}(n)$ on $M_n(\mathbb{C})$) by conjugation. (See also Example 
    \ref{example:product_type} below.)

    Related to this, Evans and Gould in \cite{EG} explicitly described the 
    $K$-theory of certain unital AF-algebras that arise from Jones' subfactors 
    \cite{J} and from WZW models associated to SU(2) as quotients of polynomial 
    rings (see also Theorem \ref{theorem:SU2k} below). In the present paper, we 
    note that braided C*-tensor categories, in the sense of Erlijman--Wenzl 
    \cite{EW}, yield a common framework for the computations of Renault, 
    Wassermann, Handelman and Rossmann on the one hand and of Evans and Gould 
    on the other by showing that in each case the ring structure in $K$-theory 
    is induced by a $*$-homomorphism, essentially defined by Erlijman and Wenzl 
    in \cite{EW}, that arises from a unitary braiding on an underlying rigid 
    C*-tensor category.
    (Note that the tensor product of Hilbert C*-modules that appears in \cite{AE} is defined via a $*$-homomorphism that is given by a similar formula.)
    For example, in Wassermann's computation the underlying category is the 
    category $\mathbf{Rep}(\mathrm{SU}(2))$ of continuous finite-dimensional 
    unitary representations of $\mathrm{SU}(2)$ whereas in the Evans--Gould 
    computation it is $\mathbf{Rep}_k(\mathrm{SU}(2))$.
    Moreover, we perform computations for the rank two Lie groups SU(3), Sp(4) 
    and G$_2$ that are analogous to the Evans--Gould computation for SU(2). Our 
    calculations are based on the work of Gepner \cite{G}, Fuchs \cite{F} and 
    others (see section \ref{subsection:fusion_rings}) in which Verlinde rings 
    were described as quotients of polynomial rings or, equivalently, of 
    representation rings.

    Let us now describe the structure of the present paper. In section \ref{section:preliminaries}, we discuss preliminary material on rigid C*-tensor categories and unitary braidings, including both the complicated examples $\mathbf{Rep}_k(G)$ and simpler examples, such as the representation categories of compact groups. We also define the fusion rings of these categories, and discuss modular $S$-matrices and the aforementioned work of Gepner, Fuchs and others on the Verlinde rings.

    Thereafter, in section \ref{section:AF-algebras}, we define the unital 
    AF-algebras with which we are concerned in the present paper. Specifically, 
    given a rigid C*-tensor category $\mathcal{C}$ and an object $\pi$ in 
    $\mathcal{C}$, we define the unital AF-algebra $A(\mathcal{C},\pi) = 
    \text{$\mathrm{ind}$-$\mathrm{lim}$}_n\mathrm{End}((\bar{\pi}\otimes\pi)^{\otimes
     n})$. Following Erlijman--Wenzl \cite{EW}, we also use a unitary braiding 
    on $\mathcal{C}$ to define a $*$-homomorphism $\theta\colon 
    A(\mathcal{C},\pi)\otimes A(\mathcal{C},\pi)\to A(\mathcal{C},\pi)$. As an 
    example, we describe in detail the situation where $\mathcal{C} = 
    \mathbf{Rep}_k(\mathrm{SU}(2))$ and $\pi$ is the fundamental 
    representation, in which case $A(\mathcal{C},\pi)$ may be described in 
    terms of Temperley--Lieb diagrams. Returning to the general setup, we next 
    show that the $*$-homomorphism $\theta$ induces a well-defined ring 
    structure on $K_0(A(\mathcal{C},\pi))$ by embedding 
    $K_0(A(\mathcal{C},\pi))$ as a subgroup of a localization of the fusion 
    ring of $\mathcal{C}$ in such a way that $K_0(\theta)$ is identified with 
    the product on the localization.

    In section \ref{section:explicit_computations}, we compute 
    $K_0(A(\mathbf{Rep}_k(G),\pi))$, where $G$ is one of the rank two Lie 
    groups SU(3), Sp(4) and G$_2$ and $\pi$ is a fundamental representation. 
    For the convenience of the reader, we also state the aforementioned 
    computations of Wassermann and Evans--Gould. The computation for $G = 
    \mathrm{SU}(3)$ is given in detail, but we omit the similar (but simpler) 
    details of the computations for Sp(4) and G$_2$. We also give a computation 
    of $K_0(A(\mathbf{Rep}(\mathrm{SU}(3)),\pi))$, but as this is essentially 
    due to Handelman and Rossmann, we only provide the main ingredients in an 
    ``elementary'' proof (using only Lie theory and general facts about ordered 
    groups and rings).

    In section \ref{section:recovering}, we show that if $\mathcal{C} = \mathbf{Rep}_k(\mathrm{SU}(2))$, where $k-1\notin 3\mathbb{Z}$ or $k\in 2\mathbb{Z}$, then one can find an object $\pi$ in $\mathcal{C}$ (or, more precisely, an explicit isomorphism class) such that $K_0(A(\mathcal{C},\pi))$ and $\mathrm{Ver}_k(\mathrm{SU}(2))$ are isomorphic as ordered rings.
    Finally, concluding remarks and open questions may be found in section \ref{section:conclusion}.

    We end the present section with a comment on notation. We will use the symbol $\mathbb{N}$ to denote the set of (strictly) positive integers. The set of non-negative integers will be denoted by $\mathbb{N}_0$.
    Given a set $X$ and $r\in\mathbb{N}$, we will denote the Cartesian product $X\times X\times\cdots\times X$ ($r$ factors) by $X^{\times r}$.

\vspace{5mm}
\noindent{\it Acknowledgements.}
We would like to thank the following entities for their generous hospitality during the writing of the present paper: The Isaac Newton Institute for Mathematical Sciences in Cambridge, England, during the research program {\it Operator Algebras: Subfactors and their Applications} in the spring of 2017; the Hausdorff Research Institute for Mathematics in Bonn, Germany, during the trimester program {\it von Neumann Algebras} in the summer of 2016; and the Dublin Institute for Advanced Studies in Dublin, Ireland, during a research visit in December 2017.

We would also like to thank the organizers of the conference {\it Young Mathematicians in C*-Algebras 2018} (YMC*A 2018) in Leuven, Belgium, the organizers of the Operator Algebra Seminar at the University of Copenhagen, and the organizers of the Analysis Seminar at Glasgow University for giving the first named author opportunities to present our work.

This research was supported by Engineering and Physical Sciences Research Council (EPSRC) grants EP/K032208/1 and EP/N022432/1.

    \section{Preliminaries}\label{section:preliminaries}

    \subsection{C*-tensor categories}\label{subsection:categories}

    In this paper, we will consider examples of so-called \emph{rigid C*-tensor categories}. A complete definition of such a category, as well as some historical background for this definition, can e.g.\ be found in Chapter 2 of \cite{NT}. (See also e.g.\ \cite{EW}, \cite{BHP}, \cite{PV}.) Here, we content ourselves with an outline of the main features of such a category. (In particular, we are sweeping under the rug the question of strictness of such a category, which is e.g.\ covered in \cite{NT}.) As we go over these features, the reader should keep in mind the following simple examples.
     \begin{itemize}
        \item The category $\mathbf{Hilb}$, whose objects are finite-dimensional Hilbert spaces and whose morphisms are linear maps.
        \item The category $\mathbf{Rep}(G)$, whose objects are finite-dimensional unitary representations of a finite group $G$ and whose morphisms are intertwining linear maps.
     \end{itemize}
For the convenience of the reader, we give a reminder about conjugation in these categories (which is one of the features of a general rigid C*-tensor category).
Given a Hilbert space $H$, the \emph{conjugate Hilbert space} $\overline{H}$ is defined, as a set, by $\overline{H} = \{\bar{\xi}\,:\,\xi\in H\}$. Its Hilbert space operations are defined by $\bar{\xi}+\bar{\eta} = \overline{\xi+\eta}$, $z\bar{\xi} = \overline{\bar{z}\xi}$ and $\langle \bar{\xi},\bar{\eta}\rangle = \langle \eta,\xi\rangle$ for $\xi,\eta\in H$ and $z\in\mathbb{C}$. Given a linear map $T\colon H\to K$, where $K$ is a Hilbert space, the \emph{conjugate operator} $\overline{T}\colon \overline{H}\to\overline{K}$ is defined by $\overline{T}(\bar{\xi}) = \overline{T(\xi)}$ for $\xi\in H$.
Given a unitary representation $\pi\colon G\to\mathcal{U}(H)$, its \emph{conjugate} (or \emph{dual}) \emph{representation} is the unitary representation $\bar{\pi}\colon G\to \mathcal{U}(\overline{H})$ defined by $\bar{\pi}(g) = \overline{\pi(g)}$ for $g\in G$. It is an easy exercise to check that, indeed, $\overline{H}$ is a Hilbert space, $\overline{T}$ is a linear map, and $\bar{\pi}$ is a unitary representation, and that if $T$ intertwines two unitary representations $\pi$ and $\rho$ then $\overline{T}$ intertwines $\bar{\pi}$ and $\bar{\rho}$.

Frobenius Reciprocity, which is mentioned below as one of the features of a rigid C*-tensor category, is a classical fact for representations of finite groups.

    \subsubsection{Features of a rigid C*-tensor category}

    A \emph{rigid C*-tensor category} $\mathcal{C}$ has the following features (where $\pi$ and $\rho$ always denote arbitrary objects of $\mathcal{C}$):

    \begin{itemize}
        \item Each morphism set $\mathrm{Mor}(\pi,\rho)$ is a (complex) Banach space.
        \item There is an involutive $*$-operation on morphisms, mapping a morphism $T\colon \pi\to\rho$ to a morphism $T^*\colon \rho\to\pi$, such that the \emph{C*-identity}
        $$
            \|T^*T\| = \|T\|^2
        $$
        holds for any morphism $T\colon \pi\to\rho$.

        Accordingly, two objects $\pi$ and $\rho$ in $\mathcal{C}$ are said to be (unitarily) \emph{isomorphic} if there exists a morphism $u\colon\pi\to\rho$ such that $u^*u = \mathrm{id}_\pi$ and $uu^* = \mathrm{id}_\rho$.

        \item One can take \emph{direct sums} of objects and of morphisms, and there is a distinguished object $\boldsymbol{0}$, which is a \emph{zero object}. In particular, $\boldsymbol{0}\oplus \pi\cong \pi\oplus\boldsymbol{0}\cong \pi$.
        \item One can take \emph{tensor products} of objects and of morphisms, and there is a distinguished object $\boldsymbol{1}$, which is a \emph{tensor unit}. In particular, $\boldsymbol{1}\otimes \pi\cong \pi\otimes\boldsymbol{1}\cong \pi$.
        \item Every object $\pi$ (resp.\ morphism $T$) has a \emph{conjugate} $\bar{\pi}$ (resp.\ $\overline{T}$).
        \item \emph{Semisimplicity}: Every object is a (finite) direct sum of simple objects. (By definition, an object $\pi$ is \emph{simple} if $\mathrm{End}(\pi)\df\mathrm{Mor}(\pi,\pi)\cong\mathbb{C}$.)

        \item The tensor unit $\boldsymbol{1}$ is a simple object.
        \item Each endomorphism space $\mathrm{End}(\pi)$ is a finite-dimensional unital C*-algebra with a canonical faithful positive trace $\mathrm{Tr}_\pi$ that satisfies $\mathrm{Tr}_{\pi\otimes\rho}(T_1\otimes T_2) = \mathrm{Tr}_\pi(T_1)\mathrm{Tr}_\rho(T_2)$ for all $T_1\in\mathrm{End}(\pi)$ and $T_2\in\mathrm{End}(\rho)$. (These traces were constructed in this general setting by Longo and Roberts \cite{LR}.)

            Moreover, the following is true. Denote by $\Lambda$ the set of isomorphism classes of simple objects in $\mathcal{C}$ and choose a representative $\pi_\lambda$ in $\lambda$ for each $\lambda\in\Lambda$. If $\pi\cong \bigoplus_{\lambda\in\Lambda} \pi_\lambda^{\oplus N_\pi^\lambda}$ for some multiplicities $N_\pi^\lambda\in\mathbb{N}_0$ then
            $$
                \mathrm{End}(\pi)\cong\bigoplus_{\lambda\in\Lambda} M_{N_\pi^\lambda}(\mathbb{C})
            $$
            as C*-algebras and $N_\pi^\lambda = \dim_{\mathbb{C}}\mathrm{Mor}(\pi_\lambda,\pi)$ for all $\lambda\in\Lambda$.
        \item Each object $\pi$ has a \emph{quantum dimension} $d(\pi) = \mathrm{Tr}_\pi(\mathrm{id}_\pi)$, which satisfies $d(\pi\otimes\rho) = d(\pi)d(\rho)$ and $d(\pi\oplus\rho) = d(\pi)+d(\rho)$ (cf.\ \cite{LR}).
        \item  \emph{Frobenius Reciprocity}, that is, certain isomorphisms\begin{gather*}\mathrm{Mor}(\pi\otimes\rho,\psi)\cong \mathrm{Mor}(\pi,\psi\otimes\bar{\rho})\cong \mathrm{Mor}(\rho,\bar{\pi}\otimes\psi).\end{gather*} In particular, $\boldsymbol{1}$ is always a direct summand of $\pi\otimes\bar{\pi}$.
    \end{itemize}
    Moreover, the operations and objects $^*$, $\bar{{\color{white}\cdot}}$, $\otimes$, $\oplus$, $\boldsymbol{0}$ and $\boldsymbol{1}$ satisfy certain natural compatibility conditions, which we will not state here.

    \begin{example}
        In addition to the examples that we mentioned above, the category $\mathbf{Rep}(G)$ of continuous finite-dimensional unitary representations of a compact (quantum) group $G$ is also a rigid C*-tensor category (cf.\ e.g.\ \cite{NT}).
    \end{example}

    Let us end this section by stating a fact that we will use later.

    \begin{remark}\label{remark:projections}
        Let $\pi$ be an object in a rigid C*-tensor category $\mathcal{C}$. Then
        the unitary equivalence classes of projections in the finite-dimensional C*-algebra $\mathrm{End}(\pi)$ are indexed by the (isomorphism classes of) direct summands $\rho$ of $\pi$. More precisely, if we denote the class corresponding to $\rho$ by $X_\rho$ then a projection $p\in \mathrm{End}(\pi)$ belongs to $X_\rho$ if and only if there exists a morphism $v\colon \rho\to\pi$ such that $v^*v = \mathrm{id}_\rho$ and $vv^* = p$. This is essentially the statement that $\mathcal{C}$ is a C*-category that admits direct sums and is semisimple.
    \end{remark}

    \subsubsection{Braided C*-tensor categories}

    The examples that we are most interested in have an additional feature, namely a unitary braiding, which we proceed to define following \cite{EW}.

    A \emph{unitary braiding} $c_{-,-}$ on a (strict) rigid C*-tensor category $\mathcal{C}$ is an assignment of an isomorphism
		$$
			c_{\pi,\rho}\colon \pi\otimes\rho\longrightarrow \rho\otimes\pi
		$$
		to every pair $(\pi,\rho)$ of objects in $\mathcal{C}$, which
		\begin{itemize}
			\item is \emph{natural}, i.e., given morphisms $T_j\colon \pi_j\to \rho_j$ ($j=1,2$), we have that $$(T_2\otimes T_1)\circ c_{\pi_1,\rho_1} =c_{\pi_2,\rho_2}\circ (T_1\otimes T_2);$$
			\item satisfies the \emph{hexagon identities}, i.e.,
			\begin{align*}
                (\id_{\rho}\otimes c_{\pi,\sigma})\circ (c_{\pi,\rho}\otimes\id_\sigma) &= c_{\pi,\rho\otimes\sigma},\\
                (c_{\pi,\sigma}\otimes\id_\rho)\circ(\id_\pi\otimes c_{\rho,\sigma}) &= c_{\pi\otimes\rho,\sigma};
            \end{align*}
			\item is \emph{unitary}, i.e., \begin{align*}
c_{\pi,\rho}^*\circ c_{\pi,\rho} &= \id_{\pi\otimes\rho},\\ c_{\pi,\rho}\circ c_{\pi,\rho}^* &= \id_{\rho\otimes\pi}.
\end{align*}
		\end{itemize}
    Following \cite{EW}, we call a rigid C*-tensor category with a distinguished unitary braiding a \emph{braided C*-tensor category}.

    \begin{example}\label{example:rep_braiding}
        The categories $\mathbf{Hilb}$ and $\mathbf{Rep}(G)$, where $G$ is a (genuine, i.e., not ``quantum'') compact group, have a unitary braiding. Namely, to a pair $(\pi,\rho)$ of unitary representations, $\pi\colon G\to\mathcal{U}(H)$ and $\rho\colon G\to\mathcal{U}(K)$, one associates the isomorphism $c_{\pi,\rho}\colon H\otimes K\to K\otimes H$ defined by $\xi\otimes\eta\mapsto \eta\otimes\xi$. This braiding is \emph{symmetric} (or \emph{trivial}) in the sense that $c_{\rho,\pi}\circ c_{\pi,\rho} = \id_{\pi\otimes\rho}$. (The categories $\mathbf{Rep}(G)$, where $G$ ranges over all compact groups, were shown to be characterized by the existence of such a braiding by Doplicher and Roberts in \cite{DR}.)
    \end{example}

    We are interested in certain rigid C*-tensor categories that have finitely many (isomorphism classes of) simple objects (i.e., they are what are often called \emph{fusion categories}) as well as an asymmetric, even non-degenerate, unitary braiding (i.e., they are what are often called \emph{unitary modular tensor categories}). Non-degeneracy of a braiding is defined as follows.
        To a rigid C*-tensor category $\mathcal{C}$ with a finite set $\Lambda$ of simple objects and a unitary braiding $c_{-,-}$, one can associate a matrix $S\in M_{\Lambda}(\mathbb{C})$, called the \emph{modular $S$-matrix}, by setting
        $$
            S_{\mu,\nu} = \mathrm{Tr}_{\mu\otimes\nu}(c_{\nu,\mu}\circ c_{\mu,\nu})
        $$
        for $\mu,\nu\in\Lambda$. The braiding $c_{-,-}$ is said to be \emph{non-degenerate}, and the category $\mathcal{C}$ to be \emph{modular}, if $S$ is an invertible matrix. (See e.g.\ \cite{Ga2} for more information.)

    \begin{example}
        For $\mathcal{C} = \mathbf{Rep}(G)$, where $G$ is a finite group, the modular $S$-matrix is given by $$S_{\mu,\nu} = \mathrm{Tr}_{\mu\otimes\nu}(\id_{\mu\otimes\nu}) = d(\mu\otimes\nu) = d(\mu)d(\nu),$$ whereby it is of rank one. Thus, this category is not modular unless $G$ is the trivial group, in which case the category in question is $\mathbf{Hilb}$.
    \end{example}

    \begin{example}\label{example:RepkG}
    It is rather more difficult to construct examples ($\neq \mathbf{Hilb}$) of modular braided C*-tensor categories. One class of examples arises from WZW models in 2d conformal field theory or, in other words, as categories of integrable highest-weight modules over certain affine Lie algebras or vertex operator algebras. (See e.g.\ \cite{HL} and the references therein.) These categories, which we will denote by $\mathbf{Rep}_k(G)$, are parameterized by pairs $(G,k)$, where $G$ is a simple, connected, simply connected, compact Lie group (such as $\mathrm{SU}(n)$ for $n\geq 2$) and $k$ is a positive integer that is known as the ``level''. We will use the symbol $G_k$ to denote the quantum group (see below) or WZW model that underlies $\mathbf{Rep}_k(G)$.

    The category $\mathbf{Rep}_k(G)$ can also be realized using level $k$ positive energy representations of the loop group $LG$ of $G$, which is the group of smooth maps $S^1\to G$ under pointwise multiplication. Note, however, that the tensor product in $\mathbf{Rep}_k(G)$ does not arise from the usual tensor product of representations
    (see \cite{PS}, \cite{Wa2}).
    Moreover, the category $\mathbf{Rep}_k(G)$ also arises from the representation theory of a certain quantum group $U_q(G)$ at a root of unity $q$, which is a Hopf algebra associated to $G$
    (see \cite{We2}). For an overview of the various realizations of $\mathbf{Rep}_k(G)$, we refer the reader to \cite{He}.

    Below, we will explain relevant aspects of these categories, but we will not provide the reader with a detailed construction. For this, we direct the reader to the references in \cite{He}. (See also Remark \ref{remark:history} below.)
    \end{example}

    \subsection{Fusion rings}\label{subsection:fusion_rings}

        Let $\mathcal{C}$ be a rigid C*-tensor category. We retain the notation of the previous section. In particular, $\Lambda$ is the set of (isomorphism classes of) simple objects in $\mathcal{C}$. Given an object $\pi$ in $\mathcal{C}$ and $\mu,\nu\in\Lambda$, we have that
	$$
		\pi\otimes \pi_\nu \cong \bigoplus_{\mu\in\Lambda} \pi_\mu^{\oplus N_{\pi,\nu}^\mu}
	$$
	for certain multiplicities $N_{\pi,\nu}^\mu = N_{\pi\otimes\nu}^\mu$, which are called the \emph{fusion rules} of the category and are organized into \emph{fusion matrices} $N_{\pi} = (N_{\pi,\nu}^\mu)_{\nu,\mu}\in M_{\Lambda}(\mathbb{C})$. It is sometimes useful to consider the \emph{fusion graph} $\Gamma_\pi$ of $\mathcal{C}$ with respect to $\pi$, whose vertex set is $\Lambda$ and whose adjacency matrix is $N_\pi$. In other words, in $\Gamma_\pi$ there are $N_{\pi,\nu}^\mu$ edges from the vertex $\nu\in\Lambda$ to the vertex $\mu\in\Lambda$.

    \begin{remark}\label{remark:PF}
        If $\Lambda$ is a finite set then the quantum dimension $d(\pi)$ of 
        $\pi$ is precisely $\|N_\pi\|$, i.e., the Perron--Frobenius eigenvalue 
        of $N_\pi$. The corresponding Perron--Frobenius eigenvector is 
        $(d(\mu))_{\mu\in\Lambda}$. (See e.g.\ \cite{Ga2} for more information.)
    \end{remark}

    \noindent The \emph{fusion ring} $F_{\mathcal{C}}$ of $\mathcal{C}$ is 
    defined as follows. As a group, it is the free abelian group 
    $\mathbb{Z}^{(\Lambda)} = \bigoplus_{\mu\in\Lambda} \mathbb{Z}\mu$, which 
    we turn into an associative unital ring by imposing the product
    $$
        \lambda\cdot\nu = \sum_{\mu\in\Lambda} N_{\lambda,\nu}^\mu\mu
    $$
    for $\lambda,\nu\in\Lambda$. (See e.g.\ \cite{F}, \cite{Ga2} for information on fusion rings.)

    Throughout this paper, we restrict attention to categories whose fusion 
    ring is commutative. Braided C*-tensor categories of course have this 
    property. However, there are many examples of rigid C*-tensor categories 
    that have a commutative fusion ring but no unitary braiding. For instance, 
    as mentioned in \cite{EW}, Drinfeld--Jimbo quantized universal enveloping 
    algebras $U_q(G)$ with $q>1$ are such examples. Moreover, it is easy to 
    find examples of rigid C*-tensor categories whose fusion ring is 
    non-commutative. For instance, given a non-abelian discrete group $\Gamma$, 
    the category $\mathbf{Rep}(\hat{\Gamma})$ is such an example. (Here, 
    $\hat{\Gamma}$ is the compact quantum group $(C_r^*(\Gamma),\Delta)$, where 
    $C_r^*(\Gamma)$ is the reduced group C*-algebra of $\Gamma$ and the 
    co-multiplication $\Delta$ is defined by $u_g\mapsto u_g\otimes u_g$ for 
    $g\in\Gamma$, where $\{u_g\}_{g\in\Gamma}$ are the canonical unitary 
    elements in $C_r^*(\Gamma)$.)

    The fusion ring of the category $\mathbf{Rep}_k(G)$ is called the 
    \emph{Verlinde ring} of $G$ at level $k$ (after Verlinde, who studied these 
    rings in \cite{Ver}) and we will denote it by $\mathrm{Ver}_k(G)$. (Here, 
    and for the remainder of this section, $G$ denotes a simple, connected, 
    simply connected, compact Lie group.) The fusion rules of 
    $\mathbf{Rep}_k(G)$ may be viewed as a truncation of the fusion rules of 
    $\mathbf{Rep}(G)$. They are given by the Ka\v{c}--Walton Formula (cf.\ 
    e.g.\ \cite{Ga2} and the references therein).

    It is a classical fact that the simple objects $\{\pi_{\vec\lambda}\}_{\vec\lambda}$ in the category $\mathbf{Rep}(G)$ are naturally parameterized by the set $\mathbb{N}_0^{\times r}$ of \emph{Dynkin labels} $\vec\lambda = (\lambda_1,\ldots,\lambda_r)$, where $r$ is the rank of $G$, in such a way that the zero vector $\vec 0$ corresponds to the trivial representation and the standard basis vectors $\vec{e}_j$ correspond to the fundamental representations. Meanwhile, the simple objects $\{\pi_{\vec\lambda}\}_{\vec\lambda}$ in the category $\mathbf{Rep}_k(G)$ are naturally indexed by a certain finite subset of $\mathbb{N}_0^{\times r}$ corresponding to integrable highest weights at level $k$. Explicitly, the Dynkin labels $\vec\lambda$ that correspond to simple objects in $\mathbf{Rep}_k(G)$ are exactly those for which the ``level'' $\ell(\vec\lambda) \df \sum_{j=1}^r a_j^\vee\lambda_j$ of $\vec\lambda$ is at most $k$. Here, $a_1^\vee,\ldots,a_r^\vee$ are the \emph{colabels} of the group $G$, which are e.g.\ explicitly given in \cite{Ga}.

    \begin{example}
        In the category $\mathbf{Rep}(\mathrm{SU}(2))$, the simple objects are indexed by the set $\Lambda = \mathbb{N}_0$ and the fusion rules are given by
        $$
            \pi_i\otimes\pi_j\cong \pi_{|i-j|}\oplus\pi_{|i-j|+2}\oplus\cdots\oplus \pi_{i+j}
        $$
        for $i,j\in\mathbb{N}_0$. Meanwhile, in the category $\mathbf{Rep}_k(\mathrm{SU}(2))$, the simple objects are indexed by the set $\Lambda = \{0,1,\ldots,k\}$ and the fusion rules are given by
    	$$
    		\pi_i\otimes\pi_j\cong    \begin{cases}
    		                              \pi_{|i-j|}\oplus\cdots\oplus\pi_{i+j}&\text{ if }i+j\leq k,\\
    		                              \pi_{|i-j|}\oplus\cdots\oplus\pi_{2k-(i+j)}&\text{ if }i+j > k,
    		                          \end{cases}
    	$$
        for $0\leq i,j\leq k$ (as shown by Gepner and Witten in \cite{GeW}).
    \end{example}

\begin{figure}
\centering
\resizebox{0.55\textwidth}{!}{%
\includegraphics{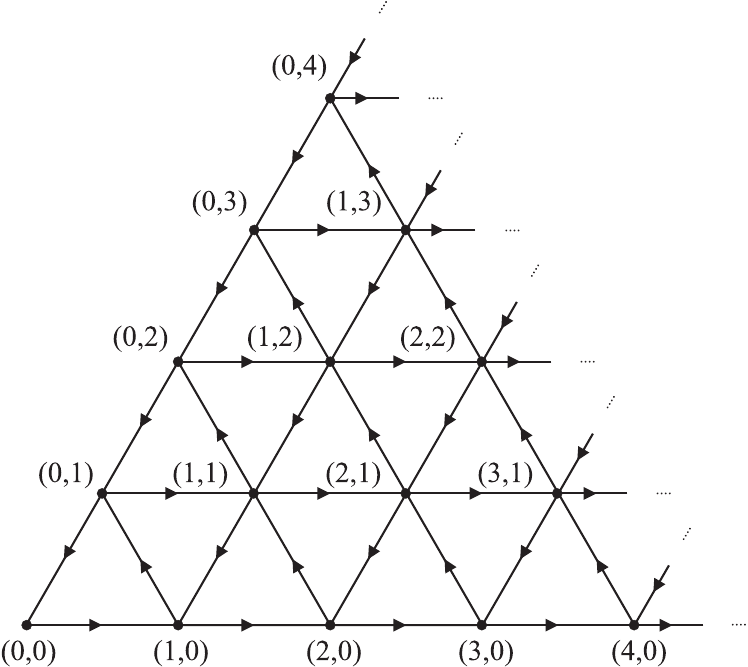}
}
\caption{\footnotesize The fusion graph of $\mathbf{Rep}(\mathrm{SU}(3))$ with
respect to $\pi_{(1,0)}$.}
\label{fig:4}
\end{figure}

    \begin{example}\label{example:SU3_fusion_rules}
        In the category $\mathbf{Rep}(\mathrm{SU}(3))$, the simple objects are indexed by the set $\Lambda = \mathbb{N}_0^{\times 2}$ and the fusion rules are determined by the fusion graph $\Gamma_{\pi_{(1,0)}}$ of $\mathbf{Rep}(\mathrm{SU}(3))$ with respect to the fundamental representation $\pi_{(1,0)}$, which is shown in figure \ref{fig:4}. Note that the fusion graph of $\mathbf{Rep}(\mathrm{SU}(3))$ with respect to $\pi_{(0,1)}=\bar{\pi}_{(1,0)}$ is obtained from $\Gamma_{\pi_{(1,0)}}$ by reversing all of the edges.

        In the category $\mathbf{Rep}_k(\mathrm{SU}(3))$, the simple objects are indexed by the set $\Lambda = \{\vec\lambda\in\mathbb{N}_0^{\times 2}\,:\,\ell(\vec\lambda)\leq k\}$, where $\ell(\vec\lambda) = \lambda_1+\lambda_2$ for $\vec\lambda = (\lambda_1,\lambda_2)\in\mathbb{N}_0^{\times 2}$. The fusion rules are determined by the fusion graph $\Gamma_{\pi_{(1,0)}}$ of $\mathbf{Rep}_k(\mathrm{SU}(3))$ with respect to the simple object $\pi_{(1,0)}$, which is obtained from the graph in figure \ref{fig:4} by discarding all vertices $\vec\lambda$ for which $\ell(\vec\lambda)>k$ (as well as any edges that are connected to them). Again, the fusion graph of $\mathbf{Rep}_k(\mathrm{SU}(3))$ with respect to $\pi_{(0,1)}=\bar{\pi}_{(1,0)}$ is obtained from $\Gamma_{\pi_{(1,0)}}$ by reversing all edges. The fusion rules of $\mathbf{Rep}_k(\mathrm{SU}(3))$ were explicitly calculated at all ``levels'' $k$ in \cite{BMW}.
    \end{example}

    In our $K$-theory computations, we will use a result of Gepner \cite{G}, which was transferred from the setting of $\mathrm{Ver}_k(G)\otimes\mathbb{C}$ (the \emph{Verlinde algebra}) to the setting of $\mathrm{Ver}_k(G)$ by Fuchs \cite{F}. To state it, we first recall some facts and introduce some terminology and notation. It is a classical theorem that when $\mathcal{C} = \mathbf{Rep}(G)$, in which case $F_{\mathcal{C}}$ is the \emph{representation ring} $R(G)$ of $G$, there is a ring isomorphism
    $
        F_{\mathcal{C}} \to \mathbb{Z}[x_1,\ldots,x_r]
    $
    that maps (the isomorphism class of) the $j$'th fundamental representation to the variable $x_j$ and the trivial representation to $1$. (As above, $r$ is the rank of $G$.) We will denote the image of $\pi_{\vec\lambda}$ under this isomorphism by $Q_{\vec\lambda}^G(x_1,\ldots,x_r)$.
    The \emph{fusion variety} $V$ associated to the category $\mathbf{Rep}_k(G)$ is the subset
    $$
        V = \{(x_{\vec\lambda}^{(1)},\ldots,x_{\vec\lambda}^{(r)})\,:\,\ell(\vec\lambda)\leq k\}
    $$
    of $\mathbb{C}^r$. In this formula, 
    $x_{\vec\lambda}^{(j)}=S_{\vec{e}_j,\vec\lambda}/S_{\vec 0,\vec\lambda}$ 
    for $j=1,\ldots,r$, where $S$ is the modular $S$-matrix for 
    $\mathbf{Rep}_k(G)$, which is given by the Ka\v{c}--Peterson Formula 
    \cite{KP}. (The formula is e.g.\ spelled out for each $G$ in \cite{Ga} and 
    for $G = \mathrm{SU}(3)$ in \cite{EP1}, but note that the formula for the 
    modular $S$-matrix on page 400 of \cite{EP1} contains a typographical 
    error: The `$+$' in the second term should be `$-$'.)
The \emph{fusion ideal} $J_k(G)$ associated to the category $\mathbf{Rep}_k(G)$ is the ideal
$$
    J_k(G) = \{p(x_1,\ldots,x_r)\in\mathbb{Z}[x_1,\ldots,x_r]\,:\,p(\vec{z}) = 0\text{ for all }\vec{z}\in V\}
$$
of $\mathbb{Z}[x_1,\ldots,x_r]$. In particular, $(x_{\vec 0}^{(1)},\ldots,x_{\vec 0}^{(r)}) = (S_{\vec{e}_1,\vec 0}/S_{\vec 0,\vec 0},\ldots,S_{\vec{e}_r,\vec 0}/S_{\vec 0,\vec 0}) = (d(\pi_{\vec{e}_1}),\ldots,d(\pi_{\vec{e}_r}))$ (cf.\ e.g.\ \cite{Ga2}) is a common zero of the polynomials in $J_k(G)$. (As usual, $d$ denotes the quantum dimension.)

    \begin{theorem}[Gepner, Fuchs]\label{theorem:Gepner}
The fusion ring of $\mathbf{Rep}_k(G)$, that is, the Verlinde ring $\mathrm{Ver}_k(G)$, is isomorphic to
$
    \mathbb{Z}[x_1,\ldots,x_r]/J_k(G)
$
in such a way that the simple object $\pi_{\vec\lambda}$ corresponds to the coset $[Q^{G}_{\vec\lambda}(x_1,\ldots,x_r)]$.
\end{theorem}

\noindent In \cite{G}, Gepner also proved that
$$
    J_k(\mathrm{SU}(2)) = \langle Q^{\mathrm{SU}(2)}_{k+1}(x_1)\rangle
$$
and
$$
    J_k(\mathrm{SU}(3)) = \langle Q^{\mathrm{SU}(3)}_{(k+1,0)}(x_1,x_2),Q^{\mathrm{SU}(3)}_{(k+2,0)}(x_1,x_2)\rangle,
$$
where, given elements $r_1,\ldots,r_n$ (resp.\ a subset $X$) in a ring $R$, we denote by $\langle r_1,\ldots,r_n\rangle$ (resp.\ $\langle X\rangle$) the (two-sided) ideal generated by the set $\{r_1,\ldots,r_n\}$ (resp.\ $X$).
Below, we will use these facts as well as the formula
$$
    J_k(\mathrm{Sp}(4)) = \langle \{Q^{\mathrm{Sp}(4)}_{\vec\lambda}(x_1,x_2)\,:\,\lambda_1+\lambda_2 = k+1\}\cup\{Q^{\mathrm{Sp}(4)}_{(0,k+2)}(x_1,x_2)\}\rangle,
$$
which follows from the work of Bourdeau--Mlawer--Riggs--Schnitzer \cite{BMRS} 
and Gepner--Schwimmer \cite{GS}, and the formula
$$
    J_k(\mathrm{G}_2) = \langle \{Q^{\mathrm{G}_2}_{\vec\lambda}(x_1,x_2)\,:\,\lambda_1+2\lambda_2 = k+1\}\cup\{Q^{\mathrm{G}_2}_{(k+2,0)}(x_1,x_2)\}\rangle,
$$
which may e.g.\ be deduced from the work of Douglas (cf.\ Theorem 1.1 in \cite{D}). For an in-depth discussion of explicit generating sets for the fusion ideals, we refer the reader to the papers \cite{BR} and \cite{D2}.

    \section{AF-algebras from braided categories}\label{section:AF-algebras}

    \subsection{Construction and basic examples}\label{subsection:construction}

    Let $\mathcal{C}$ be a rigid C*-tensor category and fix an object $\pi$ in $\mathcal{C}$. Put $\sigma = \bar{\pi}\otimes\pi$ and, for each $n\in\mathbb{N}_0$, $A(\mathcal{C},\pi)_n = \mathrm{End}(\sigma^{\otimes n})$. Define, for each $n\in\mathbb{N}_0$, a $*$-homomorphism
    $
        \iota_n\colon A(\mathcal{C},\pi)_n\to A(\mathcal{C},\pi)_{n+1}
    $
    by $\iota_n(T) = T\otimes \mathrm{id}_\sigma$ for $T\colon \sigma^{\otimes n}\to \sigma^{\otimes n}$. Since
    $$
        \mathrm{Tr}_{\sigma^{\otimes (n+1)}}(\iota_n(T)) = \mathrm{Tr}_{\sigma^{\otimes n}}(T)\mathrm{Tr}_\sigma(\mathrm{id}_\sigma) = d(\sigma)\mathrm{Tr}_{\sigma^{\otimes n}}(T)
    $$
    for all $T\colon \sigma^{\otimes n}\to \sigma^{\otimes n}$, it follows that each $\iota_n$ is injective, hence isometric. Thus, we may define the inductive limit C*-algebra
    $$
        A(\mathcal{C},\pi) = \text{$\mathrm{ind}$-$\mathrm{lim}$}_n (A(\mathcal{C},\pi)_n,\iota_n),
    $$
    which is a unital AF-algebra. Note that, up to $*$-isomorphism, $A(\mathcal{C},\pi)$ only depends on the fusion rules for $\pi$, i.e., on the matrix $N_\pi$.

    \begin{example}\label{example:product_type}
        Unsurprisingly, the easiest examples arise from the category $\mathbf{Hilb}$. If $V$ is a Hilbert space of dimension $N$ then $A(\mathbf{Hilb},V)\cong M_{N^\infty}$, where $M_{N^\infty}$ is the infinite tensor product $M_N(\mathbb{C})\otimes M_N(\mathbb{C})\otimes\cdots$, i.e., the UHF-algebra associated to the ``supernatural number'' $N^\infty$. In particular, $A(\mathbf{Hilb},\mathbb{C})\cong\mathbb{C}$ while $A(\mathbf{Hilb},\mathbb{C}^2)$ is $*$-isomorphic to the CAR-algebra.

    To get a slightly more complicated example, let $\pi$ be an object of $\mathbf{Rep}(G)$, where $G$ is a compact group. Suppose that $\pi$ is an $N$-dimensional representation. Then $A(\mathbf{Rep}(G),\pi)$ is $*$-isomorphic to the fixed point algebra $M_{N^\infty}^G$ of the following action of $G$ on $M_{N^\infty}$ by automorphisms. For each $g\in G$, view $\pi(g)$ and $\bar{\pi}(g)$ as unitary matrices in $M_N(\mathbb{C})$. Then $G$ acts via the formula
    $$
        g\cdot (x_1\otimes x_2\otimes \cdots) = \bar{\pi}(g)x_1\bar{\pi}(g)^*\otimes \pi(g)x_2\pi(g)^*\otimes\cdots.
    $$
    As we have mentioned before, such fixed point algebras were much studied by Wassermann \cite{Wa}, Handelman and Rossmann \cite{HR1,HR2}, and Handelman \cite{H1,H2} in the early 1980s, and some of their $K$-theory computations will be presented later in this paper (cf.\ Theorem \ref{theorem:SU2}, Theorem \ref{theorem:SU3} and Corollary \ref{corollary:SO3}).
    \end{example}

    \begin{remark}
        As e.g.\ explained in Remark 8.2 of \cite{PV}, it follows from the work of many people, including Jones \cite{J}, Wenzl \cite{We1,We2}, Popa \cite{Po}, Banica \cite{Ba} and Xu \cite{Xu}, that $A(\mathcal{C},\pi)$ is the inductive limit of a tower of higher relative commutants arising from a certain subfactor associated to the pair $(\mathcal{C},\pi)$.
    \end{remark}

    \subsection{From unitary braidings to $*$-homomorphisms}\label{subsection:homomorphisms}

    We will next, following \cite{EW}, define a $*$-homomorphism $A(\mathcal{C},\pi)\otimes A(\mathcal{C},\pi)\to A(\mathcal{C},\pi)$ that will turn out to induce a multiplication map on $K_0(A(\mathcal{C},\pi))$.

    \subsubsection{Definition of the $*$-homomorphism}

    Assume now moreover that $\mathcal{C}$ is a braided C*-tensor category with unitary braiding $c_{-,-}$. For each $n\in\mathbb{N}$, Erlijman and Wenzl in \cite{EW} defined a $*$-homomorphism
    $$
        \theta_n\colon A(\mathcal{C},\pi)_n\otimes A(\mathcal{C},\pi)_n\longrightarrow A(\mathcal{C},\pi)_{2n}
    $$
    by
    $$
        \theta_n(T_1\otimes T_2) = U_n(T_1\otimes T_2)U_n^*
    $$
    for $T_1,T_2\in A(\mathcal{C},\pi)_n = \mathrm{End}(\sigma^{\otimes n})$, where
    $$
        U_n = \prod_{j=1}^{n-1}\big(\id_\sigma^{\otimes (n-j)}\otimes c_{\sigma,\sigma}^{\otimes j}\otimes\id_\sigma^{\otimes (n-j)}\big)\in\mathcal{U}(A(\mathcal{C},\pi)_{2n}).
    $$
    Using a well-known graphical calculus for braided C*-tensor categories, $\theta_n(T_1\otimes T_2)$ may be depicted as follows when $n=3$.
    \begin{center}
    \resizebox{.55\textwidth}{!}{%
    \begin{tikzpicture}[scale = .5]
        \node() at (-3.5,0) {\Large $\theta_n(T_1\otimes T_2) =$};
        \begin{scope}[ultra thick]
            %morphisms
            \draw (0.5,-1.5) rectangle (3.5,1.5);
            \draw (4.5,-1.5) rectangle (7.5,1.5);
            \node() at (2,0) {\Huge $T_1$};
            \node() at (6,0) {\Huge $T_2$};

            \begin{scope}[color=blue,dashed]
                \draw (1,1.5) -- (1,6); \draw (2,1.5) -- (2,4); \draw (3,1.5) -- (3,2);
                \draw (1,-1.5) -- (1,-6); \draw (2,-1.5) -- (2,-4); \draw (3,-1.5) -- (3,-2);

                \draw (3,2) .. controls (3.5,3.5) and (4.5,2.5) .. (5,4);
                \draw (2,4) .. controls (2.25,5.5) and (2.75,4.5) .. (3,6);
                \draw (5,4) .. controls (5.25,5.5) and (5.75,4.5) .. (6,6);

                \draw (3,-2) .. controls (3.5,-3.5) and (4.5,-2.5) .. (5,-4);
                \draw (2,-4) .. controls (2.25,-5.5) and (2.75,-4.5) .. (3,-6);
                \draw (5,-4) .. controls (5.25,-5.5) and (5.75,-4.5) .. (6,-6);
            \end{scope}

            \begin{scope}[color=red]
                \draw (5,1.5) -- (5,2);\draw (6,1.5) -- (6,4); \draw (7,1.5) -- (7,6);
                \draw (5,-1.5) -- (5,-2);\draw (6,-1.5) -- (6,-4); \draw (7,-1.5) -- (7,-6);

                \draw (3,4) .. controls (3.5,2.5) and (4.5,3.5) .. (5,2);
                \draw (2,6) .. controls (2.25,4.5) and (2.75,5.5) .. (3,4);
                \draw (5,6) .. controls (5.25,4.5) and (5.75,5.5) .. (6,4);

                \draw (3,-4) .. controls (3.5,-2.5) and (4.5,-3.5) .. (5,-2);
                \draw (2,-6) .. controls (2.25,-4.5) and (2.75,-5.5) .. (3,-4);
                \draw (5,-6) .. controls (5.25,-4.5) and (5.75,-5.5) .. (6,-4);
            \end{scope}

            \begin{scope}[thin, dashed]
                \draw (0,-2) rectangle (8,2);
                \draw (0,2) rectangle (8,4);
                \draw (0,4) rectangle (8,6);
                \draw (0,-4) rectangle (8,-2);
                \draw (0,-6) rectangle (8,-4);
            \end{scope}		
        \end{scope}
	\end{tikzpicture}	
}
\end{center}
The following lemma was proved in \cite{EW}.
\begin{lemma}[Erlijman--Wenzl]\label{lemma:EW} For each $n\in\mathbb{N}$,
    $$
        \theta_{n+1}\circ (\iota_n\otimes\iota_n) = \iota_{2n+1}\circ\iota_{2n}\circ \theta_n.
    $$
\end{lemma}
\noindent In fact, Erlijman and Wenzl used the maps $\theta_n$ to construct new examples of subfactors. In the present paper, we only use the above lemma (and the fact that each $\theta_n$ is isometric) to get an induced $*$-homomorphism
$$
    \theta\colon A(\mathcal{C},\pi)\otimes A(\mathcal{C},\pi)\to A(\mathcal{C},\pi).
$$
For the convenience of the reader, we give an algebraic proof of the lemma.

\begin{proof}[Proof of Lemma \ref{lemma:EW}] (For simplicity, we give the proof in the case where $\mathcal{C}$ is a strict category, but this is not an actual restriction.)
    Note first that
    $$
        c_{\sigma,\sigma^n} = (\id_{\sigma}^{\otimes (n-1)}\otimes c_{\sigma,\sigma})(c_{\sigma,\sigma^{n-1}}\otimes\id_\sigma)
    $$
    for all $n\geq 1$ by the hexagon identities.

    We next claim that
    \begin{equation}\label{equation:braid}
        U_{n+1} = (U_n\otimes\id_\sigma^{\otimes 2})(\id_\sigma^{\otimes n}\otimes c_{\sigma,\sigma^n}\otimes\id_\sigma)
    \end{equation}
    for all $n\geq 1$. The statement is clear for $n = 1$. Assume that the statement holds for $n = k-1$. Then, writing $\id$ for $\id_\sigma$ and using the fact that
    \begin{equation*}
        U_{n+1} = \id\otimes c_{\sigma,\sigma}^{\otimes n}U_n\otimes\id
    \end{equation*}
    for all $n\geq 1$ by definition, we get that
    \begin{align*}
        U_{k+1} &= \id\otimes c_{\sigma,\sigma}^{\otimes k}U_k\otimes\id\\
        &= \id\otimes \bigg(c_{\sigma,\sigma}^{\otimes k}(U_{k-1}\otimes\id^{\otimes 2})(\id^{\otimes (k-1)}\otimes c_{\sigma,\sigma^{k-1}}\otimes\id) \bigg)\otimes\id\\
        &= \id\otimes \bigg(\big((c_{\sigma,\sigma}^{\otimes (k-1)}U_{k-1})\otimes\id^{\otimes 2}\big)\\
        &\qquad\qquad\qquad\circ(\id^{\otimes(2k-2)}\otimes c_{\sigma,\sigma})(\id^{\otimes (k-1)}\otimes c_{\sigma,\sigma^{k-1}}\otimes\id)\bigg)\otimes\id\\
        &=\id\otimes \bigg(\big((c_{\sigma,\sigma}^{\otimes (k-1)}U_{k-1})\otimes\id^{\otimes 2}\big)\big(\id^{\otimes(k-1)}\otimes c_{\sigma,\sigma^k}\big)\bigg)\otimes\id\\
        &= ((\id\otimes c_{\sigma,\sigma}^{\otimes (k-1)}U_{k-1}\otimes\id)\otimes\id^{\otimes 2})(\id^{\otimes k}\otimes c_{\sigma,\sigma^k}\otimes\id)\\
        &= (U_k\otimes\id^{\otimes 2})(\id^{\otimes k}\otimes c_{\sigma,\sigma^k}\otimes\id),
    \end{align*}
    which is exactly the statement for $n = k$.

    We are now ready to prove that
    $$
        \theta_{n+1}\circ (\iota_n\otimes\iota_n) = \iota_{2n+1}\circ\iota_{2n}\circ \theta_n
    $$
    for all $n\geq 1$. Setting $\gamma = \big(\theta_{n+1}\circ (\iota_n\otimes\iota_n)\big)(T_1\otimes T_2)$), equation (\ref{equation:braid}) and the naturality of the braiding imply that
    \begin{align*}
        \gamma &= U_{n+1}(T_1\otimes\id\otimes T_2\otimes\id)U_{n+1}^{-1}\\
        &= (U_n\otimes\id^{\otimes 2})(T_1\otimes (c_{\sigma,\sigma^n}(\id\otimes T_2)c_{\sigma,\sigma^n}^{-1})\otimes\id)(U_n^{-1}\otimes\id^{\otimes 2})\\
        &= (U_n\otimes\id^{\otimes 2})(T_1\otimes (T_2\otimes\id)\otimes\id)(U_n^{-1}\otimes\id^{\otimes 2})\\
        &= (\iota_{2n+1}\circ\iota_{2n}\circ \theta_n)(T_1\otimes T_2),
    \end{align*}
    which proves what we wanted.
\end{proof}

\begin{example}\label{example:rep_homom}
    Let $\mathcal{C}$ be the category $\mathbf{Rep}(G)$, where $G$ is a compact group. Then, as we saw in Example \ref{example:product_type}, $A(\mathbf{Rep}(G),\pi)$ is $*$-isomorphic to a fixed point algebra $M_{N^\infty}^G$, where $N$ is the dimension of $\pi$. Under this identification (recalling the unitary braiding on $\mathbf{Rep}(G)$ from Example \ref{example:rep_braiding}), the above $*$-homomorphism $\theta\colon A(\mathcal{C},\pi)\otimes A(\mathcal{C},\pi)\to A(\mathcal{C},\pi)$ corresponds to the restriction $M_{N^\infty}^G\otimes M_{N^\infty}^G\to M_{N^\infty}^G$ of the $*$-isomorphism $M_{N^\infty}\otimes M_{N^\infty}\to M_{N^\infty}$ that interlaces the tensor factors.
\end{example}

\subsubsection{SU(2) and Temperley--Lieb--Jones 
algebras}\label{subsubsection:TL}

We will next explain how the $*$-homomorphism $\theta$ is defined when $\mathcal{C}\! =\! \mathbf{Rep}_k(\mathrm{SU}(2))$ and $\pi \!=\! \pi_1$.
We recall first the notion of an \emph{$(m,n)$-Temperley--Lieb diagram} (for 
$m,n\in\mathbb{N}_0$
of equal parity), which first appeared in the work of Kauffman \cite{Kau}. Such 
a diagram consists of $(n+m)/2$ non-crossing smooth strands inside a rectangle 
with $m$ marked points on the upper edge and $n$ marked points on the lower 
edge, each marked point being connected to a unique strand. For example, the 
following is a $(4,6)$-Temperley--Lieb diagram.
    \begin{center}
\resizebox{.35\textwidth}{!}
{
\begin{tikzpicture}
    \draw (0,0) rectangle (7,3);
    \draw[ultra thick,color=blue] (1,0) to [out=90, in=195] (2.67,1.5);
    \draw[ultra thick,color=blue] (2.67,1.5) to [out=15,in=-90] (4.33,3);
    \draw[ultra thick,color=blue] (1,3) to [out=270,in=180] (1.83,2.2);
        \draw[ultra thick,color=blue] (1.83,2.2) to [out=0,in=270] (2.67,3);
    \draw[ultra thick,color=blue] (6,0) -- (6,3);
    \draw[ultra thick,color=blue] (2,0) to [out=90,in=180] (3.5,1);
    	\draw[ultra thick,color=blue] (3.5,1) to [out=0,in=90] (5,0);
    \draw[ultra thick,color=blue] (3,0) to [out=90,in=180] (3.5,0.5);
        	\draw[ultra thick,color=blue] (3.5,0.5) to [out=0,in=90] (4,0);
\end{tikzpicture}
}
\end{center}
We next use Temperley-Lieb diagrams to define some (complex) algebras. As a 
vector space, the \emph{$n$'th Temperley--Lieb algebra with parameter 
$\delta\in\mathbb{C}$} (for $n\in\mathbb{N}_0$) is the formal (complex) linear 
span of all $(n,n)$-Temperley--Lieb diagrams, i.e.,
    $$
        \mathrm{TL}_{n}(\delta) = 
        \mathrm{span}_{\mathbb{C}}\{\text{$(n,n)$-Temperley--Lieb diagrams}\}.
    $$
    We define the product of two diagrams by stacking them, aligning marked 
    points, smoothing strands, and removing all closed loops at the cost of 
    multiplying by $\delta^N$, where $N$ is the total number of closed loops. 
    The following picture illustrates the product of two 
    $(4,4)$-Temperley--Lieb diagrams in the algebra $\mathrm{TL}_{4}(\delta)$.
    \begin{center}
    \resizebox{.55\textwidth}{!}
    {
    \begin{tikzpicture}
    \draw (0,0) rectangle (5,4);
    \draw[ultra thick,color=blue] (1,0) -- (1,2);
	\draw[ultra thick,color=blue] (3,0) to [out=90,in=180] (3.5,0.5);
	\draw[ultra thick,color=blue] (3.5,0.5) to [out=0,in=90] (4,0);
    \draw[ultra thick,color=blue] (2.5,2) circle [radius=0.5];
    \draw[ultra thick,color=blue] (1,2) to [out=90,in=90] (4,2);
    \draw[ultra thick,color=blue] (1,4) to [out=270,in=180] (1.5,3.5);
    \draw[ultra thick,color=blue] (1.5,3.5) to [out=0,in=270] (2,4);
    \draw[ultra thick,color=blue] (3,4) to [out=270,in=180] (3.5,3.5);
    \draw[ultra thick,color=blue] (3.5,3.5) to [out=0,in=270] (4,4);
    \draw[ultra thick,color=blue] (2,0) to [out=90, in=195] (3,1);
    \draw[ultra thick,color=blue] (3,1) to [out=15,in=-90] (4,2);
    \draw[dashed] (0,2) -- (5,2);
    \node() at (6.5,2) {\huge $=\,\,\delta$};
    \draw (7.5,1) rectangle (12.5,3);
    \draw[ultra thick,color=blue] (8.5,3) to [out=270,in=180] (9,2.5);
    \draw[ultra thick,color=blue] (9,2.5) to [out=0,in=270] (9.5,3);
    \draw[ultra thick,color=blue] (10.5,3) to [out=270,in=180] (11,2.5);
    \draw[ultra thick,color=blue] (11,2.5) to [out=0,in=270] (11.5,3);
    \draw[ultra thick,color=blue] (8.5,1) to [out=90,in=180] (9,1.5);
    \draw[ultra thick,color=blue] (9,1.5) to [out=0,in=90] (9.5,1);
    \draw[ultra thick,color=blue] (10.5,1) to [out=90,in=180] (11,1.5);
    \draw[ultra thick,color=blue] (11,1.5) to [out=0,in=90] (11.5,1);
    \end{tikzpicture}
    }
    \end{center}
The \emph{$n$'th Temperley--Lieb--Jones algebra with parameter $\delta$} is 
defined as
    $$
        \widetilde{\mathrm{TL}}_{n}(\delta) = \mathrm{TL}_{n}(\delta)/\{x\,:\,\mathrm{Tr}(xy) = 0\text{ for all }y\},
    $$
    where the trace $\mathrm{Tr}$ (often called a \emph{Markov trace}) is defined on diagrams by
\begin{center}
\resizebox{.3\textwidth}{!}{%
\includegraphics{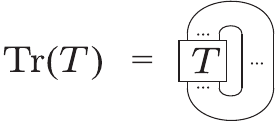}
}
\end{center}
Note that $\widetilde{\mathrm{TL}}_{n}(\delta)$ is a $*$-algebra under the $*$-operation that reflects diagrams about a horizontal axis and that the map $i_n\colon \widetilde{\mathrm{TL}}_{n}(\delta)\to \widetilde{\mathrm{TL}}_{n+1}(\delta)$, defined on diagrams by placing a vertical strand on the right side, is an injective $*$-homomorphism.

We believe that the following result may be ascribed to Jones \cite{J}, 
Kauffman \cite{Kau} and Goodman--de la Harpe--Jones \cite{GHJ}.

\begin{theorem}[Jones, Kauffman, Goodman--de la Harpe--Jones] Let 
$k\in\mathbb{N}$ be given and put $\delta = 2\cos(\pi/(k+2))$. Then 
$\widetilde{\mathrm{TL}}_{n}(\delta)$ is a C*-algebra for all $n$, and there is 
a $*$-isomorphism
    $$
        A(\mathbf{Rep}_k(\mathrm{SU}(2)),\pi_1)\cong \text{$\mathrm{ind}$-$\mathrm{lim}$}_n (\widetilde{\mathrm{TL}}_{n}(\delta),i_n).
    $$
\end{theorem}

\begin{remark}
    Put $\mathscr{J} = \{2\cos(\pi/(k+2))\,:\,k\in\mathbb{N}\}\cup[2,\infty)$. A version of Jones' Index Rigidity Theorem from \cite{J} states that, given a parameter $\delta$, $\mathrm{Tr}$ is a positive linear functional on $\mathrm{TL}_n(\delta)$ for all $n$ if and only if $\delta\in \mathscr{J}$. In particular, $\widetilde{\mathrm{TL}}_{n}(\delta)$ is a C*-algebra for all $n$ whenever $\delta\in \mathscr{J}$.
\end{remark}

    Under the identification in the preceding theorem, $\theta$ is given by superposition of diagrams. For instance, if we define the $\theta_n$ and $U_n$ in terms of the self-conjugate object $\pi = \pi_1$ rather than $\sigma$ then  $\theta_3$ may be depicted in the following way.
    \begin{center}
\resizebox{.8\textwidth}{!}{%
\begin{tikzpicture}[scale = .9]
	   \begin{scope}[ultra thick]
	       \draw[thin, dashed] (-.2,0) rectangle (1.6,2);
	       \draw[thin, dashed] (2.4,0) rectangle (4.2,2);
	       \draw[thin, dashed] (5,0) rectangle (8.6,2);

		   \draw[color=blue] (0,2) .. controls (0,1) and (1.4,1) .. (1.4,0);
				\draw[color=blue] (0,0) .. controls (.45,.8) .. (.9,0);
				\draw[color=blue] (.5,2) .. controls (.95,1) .. (1.4,2);

				\node() at (2,1) {$\bigotimes$};

				\draw[color=red] (2.7,0) -- (2.7,2);
				\draw[color=red] (3,0) .. controls (3.45,1) .. (3.9,0);
				\draw[color=red] (3,2) .. controls (3.45,1) .. (3.9,2);

				\node() at (4.6,1) {$\mapsto$};

				\draw[dashed, color=blue] (5.3,2) .. controls (5.3,1) and (7.9,1) .. (7.9,0);
				\draw[dashed, color=blue] (5.3,0) .. controls (5.95,1) .. (6.6,0);
				\draw[dashed, color=blue] (6.6,2) .. controls (7.35,1) .. (7.9,2);
				\draw[color=red] (5.95,0) -- (5.95,2);
				\draw[color=red] (7.05,0) .. controls (7.7,1) .. (8.35,0);
				\draw[color=red] (7.05,2) .. controls (7.7,1) .. (8.35,2);
				\node() at (9,1) {$=$};

				\draw[thin, dashed] (9.4,2) rectangle (13,4);
				\draw[thin, dashed] (9.4,0) rectangle (13,2);
				\draw[thin, dashed] (9.4,-2) rectangle (13,0);

				\draw[color=blue, dashed] (9.7,2) .. controls (9.7,1) and (11,1) .. (11,0);
				\draw[color=blue, dashed] (9.7,0) .. controls (10.15,.8) .. (10.6,0);
				\draw[color=blue, dashed] (10.1,2) .. controls (10.55,1) .. (11,2);
				\draw[color=red] (11.4,0) -- (11.4,2);
				\draw[color=red] (11.8,0) .. controls (12.25,1) .. (12.7,0);
				\draw[color=red] (11.8,2) .. controls (12.25,1) .. (12.7,2);

                \draw[color=blue, dashed] (9.7,2) -- (9.7,4);
                \draw[color=blue, dashed] (10.1,2) .. controls (10.1,2.8) and (11,3) .. (11,4);
                \draw[color=blue, dashed] (11,2) .. controls (11,2.2) and (11.8,3) .. (11.8,4);
                \draw[color=red] (11.4,2) .. controls (11.4,3) and (10.1,3) .. (10.1,4);
                \draw[color=red] (11.8,2) .. controls (11.6,2.5) and (11.4,3.5) .. (11.4,4);
                \draw[color=red] (12.7,2) -- (12.7,4);

                \draw[color=blue, dashed] (9.7,0) -- (9.7,-2);
                \draw[color=blue, dashed] (10.6,0) .. controls (11,-1) .. (11,-2);
                \draw[color=blue, dashed] (11,0) .. controls (11,-.5) and (11.8,-1) .. (11.8,-2);
                \draw[color=red] (11.4,0) .. controls (11.4,-1) and (10.1,-1) .. (10.1,-2);
                \draw[color=red] (11.8,0) .. controls (11.6,-.5) and (11.4,-1.5) .. (11.4,-2);
                \draw[color=red] (12.7,0) -- (12.7,-2);
	   \end{scope}		
	\end{tikzpicture}	
}
\end{center}
Here, the diagram in the top-most rectangle on the right represents the unitary element $U_3$ and the crossings are to be interpreted as follows.
\begin{center}
    \resizebox{.75\textwidth}{!}{%
\begin{tikzpicture}
      \draw[dashed] (0,1) -- (1,0);
      \draw (0,0) -- (1,1);
      \node() at (1.8,.5) {$\,\,= ie^{\frac{\pi i}{2(k+2)}}$};
      \draw (2.5,0) .. controls (3,.5) .. (2.5,1);
      \draw (3.5,0) .. controls (3,.5) .. (3.5,1);
      \node() at (4.3,.5) {$\,-\,ie^{-\frac{\pi i}{2(k+2)}}$};
      \draw (5.08,1) .. controls (5.58,.5) .. (6.08,1);
      \draw (5.08,0) .. controls (5.58,.5) .. (6.08,0);
%      \node() at (9.2,.5) {$\in\mathcal{U}(\widetilde{\mathrm{TL}}_{2}(\delta_k))$};
\end{tikzpicture}
}
\end{center}
This formula for a crossing arose from Kauffman's work on knot invariants \cite{Kau}.

\begin{remark}\label{remark:history}
     The Temperley--Lieb algebras first appeared in the work of Temperley and 
     Lieb \cite{TL}
    on Potts and ice-type models
    in statistical mechanics, in which they were defined in terms of generators and relations.
    These relations reappeared in the work of Jones \cite{J}, in which the 
    Temperley--Lieb--Jones algebras manifested as subalgebras of higher 
    relative commutants of subfactors (see also \cite{GHJ}).
    Soon, in a paper \cite{Kau} about a knot invariant introduced by Jones 
    \cite{J85}, Kauffman described Temperley--Lieb algebras in terms of 
    Temperley--Lieb diagrams (see also \cite{JR}).

    Later, it was realized that a diagrammatic description could also be given 
    for standard invariants of subfactors (cf.\ Jones' introduction of 
    subfactor planar algebras \cite{J2} based on work of Popa \cite{Po}) and 
    tensor categories (cf.\ \cite{Tur} as well as e.g.\ \cite{EW}, \cite{BHP}). 
    For example, the category $\mathbf{Rep}_k(\mathrm{SU}(2))$ was described in 
    terms of Temperley--Lieb diagrams (cf.\ \cite{Tur} as well as e.g.\ 
    \cite{Yam}, \cite{Coo}, \cite{EP12}) while the category 
    $\mathbf{Rep}_k(\mathrm{SU}(3))$ was described in terms of so-called 
    A$_2$-Temperley--Lieb diagrams (cf.\ \cite{Ku}, in which 
    A$_2$-Temperley--Lieb diagrams were first introduced, as well as e.g.\ 
    \cite{We1}, \cite{Coo}, \cite{EP0}, \cite{EP12}). In fact, it was observed 
    that rigid C*-tensor categories are in some sense the same as so-called 
    factor planar algebras (cf.\ \cite{MPS}, \cite{BHP}).
\end{remark}

    \subsection{Induced (ordered) ring structure in $K$-theory}\label{subsection:ring_structure}

    Using the ``external product'' $K_0(A(\mathcal{C},\pi))\otimes_{\mathbb{Z}}K_0(A(\mathcal{C},\pi))\to K_0(A(\mathcal{C},\pi)\otimes A(\mathcal{C},\pi))$ in $K$-theory (cf.\ e.g.\ section 4.7 of \cite{HiR}), which is in fact a group isomorphism in this case, the $*$-homomorphism $\theta$ induces a binary operation
    $$
        K_0(\theta)\colon K_0(A(\mathcal{C},\pi))\otimes_{\mathbb{Z}}K_0(A(\mathcal{C},\pi))\longrightarrow K_0(A(\mathcal{C},\pi)),
    $$
    which will turn out to be an associative product on $K_0(A(\mathcal{C},\pi))$ that endows it with the structure of a unital ring. We will prove this by identifying $K_0(A(\mathcal{C},\pi))$ with a subgroup of the localization $F_{\mathcal{C}}[\sigma^{-1}]$ in such a way that $K_0(\theta)$ corresponds to the product on $F_{\mathcal{C}}[\sigma^{-1}]$. (Recall that $F_{\mathcal{C}}[\sigma^{-1}]$ consists of equivalence classes of formal fractions $x/\sigma^n$ with $x\in F_{\mathcal{C}}$ and $n\in\mathbb{N}_0$ under the equivalence relation $\sim$ defined as follows: $x/\sigma^n\sim y/\sigma^m$ if $\sigma^N(\sigma^mx-\sigma^ny) = 0$ in $F_{\mathcal{C}}$ for some $N\in\mathbb{N}_0$.)

    Identifying $K_0(A(\mathcal{C},\pi)_n)$ with 
    $\bigoplus_{\mu\prec\sigma^{\otimes n}}\mathbb{Z}\mu$, where 
    $\mu\prec\sigma^{\otimes n}$ is to be read as\linebreak ``$\mu$ occurs as a 
    direct summand of $\sigma^{\otimes n}$'', we have that 
    $K_0(A(\mathcal{C},\pi))$ is isomorphic (as an ordered group) to the limit 
    of the inductive sequence
    \begin{equation}\label{equation:system}
        \cdots \longrightarrow \bigoplus_{\mu\prec\sigma^{\otimes n}}\mathbb{Z}\mu\overset{M_\sigma}{\longrightarrow}\bigoplus_{\mu\prec\sigma^{\otimes(n+1)}}\mathbb{Z}\mu\longrightarrow\cdots,
    \end{equation}
    where $M_\sigma$ is defined on basis vectors by
    $$
        M_\sigma(\mu) = \sum_{\nu\in\Lambda}N_{\sigma,\mu}^{\nu}\nu
    $$
    and the positive cone in $\bigoplus_{\mu\prec\sigma^{\otimes n}}\mathbb{Z}\mu$ is $\{\sum_{\mu\prec\sigma^{\otimes n}}v_\mu\mu\,:\,v_\mu\geq 0\text{ for all }\mu\}$. Note that if we view $\bigoplus_{\mu\prec\sigma^{\otimes n}}\mathbb{Z}\mu$ as a subset of $F_{\mathcal{C}}$ then $M_{\sigma}$ just multiplies by $\sigma$.
    Define, for each $n\in\mathbb{N}_0$, a group homomorphism
    $$
        \phi_n\colon \bigoplus_{\mu\prec\sigma^{\otimes n}}\mathbb{Z}\mu\longrightarrow F_{\mathcal{C}}[\sigma^{-1}]
    $$
    by $\phi_n(\mu) = \mu/\sigma^n$ for $\mu\prec\sigma^{\otimes n}$. Then clearly $\phi_{n+1}\circ M_\sigma = \phi_n$ for all $n$, whereby we get an induced group homomorphism
    $$
        \phi\colon K_0(A(\mathcal{C},\pi))\longrightarrow F_{\mathcal{C}}[\sigma^{-1}].
    $$
    It is straightforward to verify that $\phi$ is injective. We claim that $m\circ (\phi_n\otimes\phi_n) = \phi_{2n}\circ K_0(\theta_n)$, i.e., that the following diagram commutes, for all $n$, where $m$ is the product on $F_{\mathcal{C}}[\sigma^{-1}]$.
     $$
    	\begin{CD}
    	            K_0(\mathrm{End}(\sigma^{\otimes n}))\otimes_{\mathbb{Z}}
    	            K_0(\mathrm{End}(\sigma^{\otimes n}))
    	            @>\phi_n\otimes\phi_n>> F_{\mathcal{C}}[\sigma^{-1}]
    	            \otimes_{\mathbb{Z}} F_{\mathcal{C}}[\sigma^{-1}] \\
    	            @VK_0(\theta_n) VV @VVm V\\
    	            K_0(\mathrm{End}(\sigma^{\otimes 2n}))
    	            @>\phi_{2n}>>F_{\mathcal{C}}[\sigma^{-1}]
    	        \end{CD}
    $$
    By Remark \ref{remark:projections}, a projection $p\in \mathrm{End}(\sigma^{\otimes n})$ belongs to the (unitary equivalence) class corresponding to the direct summand $\mu$ of $\sigma^{\otimes n}$ if and only if there exists an
    element $v\in\mathrm{Mor}(\mu,\sigma^{\otimes n})$ such that $v^*v = \id_\mu$ and $vv^* = p$. It follows that if $p\in \mathrm{End}(\sigma^{\otimes n})$ belongs to the class corresponding to $\mu$ and $q\in \mathrm{End}(\sigma^{\otimes n})$ to the one corresponding to $\nu$ then $\theta_n(p\otimes q) = U_n(p\otimes q)U_n^*\in \mathrm{End}(\sigma^{\otimes 2n})$ belongs to the class corresponding to the direct summand $\mu\otimes\nu$ of $\sigma^{\otimes 2n}$. This proves the claim. We conclude that $m\circ (\phi\otimes\phi) = \phi\circ K_0(\theta)$, whereby $K_0(\theta)$ has the properties that we asserted at the beginning of this section.

    \begin{remark}\label{remark:subcategory}
        Denote by $\mathcal{C}_1$ the full C*-tensor subcategory of $\mathcal{C}$ generated by the simple objects that occur as direct summands in tensor powers of $\sigma$. Then $\phi$ maps into the subring $\mathrm{F}_{\mathcal{C}_1}[\sigma^{-1}]$ of $\mathrm{F}_{\mathcal{C}}[\sigma^{-1}]$.
    \end{remark}

    By an \emph{ordered ring}\label{page:ordered_ring}, we shall mean a commutative ring $R$ with identity equipped with a positive cone $R_+$, i.e., a subset $R_+\subset R$ satisfying $R_++R_+\subset R_+$, $R_+\cap (-R_+) = \{0\}$ and $R = R_+-R_+$, for which $R_+\cdot R_+\subset R_+$ and in which the identity element $1$ is an order unit for the ordered group $(R,R_+)$ under addition (cf.\ Remark \ref{remark:SU3}).
    In particular, we have the following. Define a positive cone in $F_{\mathcal{C}}[\sigma^{-1}]$ by
    $$
        \mathrm{F}_{\mathcal{C}}[\sigma^{-1}]_+ = \{x\,:\,d(x)>0\}\cup\{0\},
    $$
    where $d$ denotes the ring homomorphism $d\colon \mathrm{F}_{\mathcal{C}}[\sigma^{-1}]\to\mathbb{R}$ defined by $d(\mu/\sigma^n) = d(\mu)/d(\sigma)^n$ for $\mu\in\Lambda$ and $n\in\mathbb{N}_0$. (On the right hand side, $d$ denotes the quantum dimension in $\mathcal{C}$.) Then $K_0(A(\mathcal{C},\pi))$ and $\mathrm{F}_{\mathcal{C}}[\sigma^{-1}]$ are ordered rings and $\phi$ is a positive map, i.e., $\phi(K_0(A(\mathcal{C},\pi))_+)\subset \mathrm{F}_{\mathcal{C}}[\sigma^{-1}]_+$.

    \begin{remark}\label{remark:ordered_ring_iso}
        Assume now that $\mathcal{C}$ has only finitely many simple objects. We claim that $\phi$ is an ordered ring isomorphism onto the ordered subring $\mathrm{F}_{\mathcal{C}_1}[\sigma^{-1}]$ of $\mathrm{F}_{\mathcal{C}}[\sigma^{-1}]$. This follows from the fact that, in this case, the directed system in equation (\ref{equation:system}) is \emph{stationary} in the sense of Chapter 6 in \cite{E}, and is closely related to the fact that, by Theorem 6.1 in \cite{E}, $A(\mathcal{C},\pi)$ has a unique (normalized) faithful positive trace. (Note also that $A(\mathcal{C},\pi)$ is a simple C*-algebra.)

        Let us provide some further details. It is easy to show that $\phi$ is a ring isomorphism onto $\mathrm{F}_{\mathcal{C}_1}[\sigma^{-1}]$. (In particular, $K_0(A(\mathcal{C},\pi))$ is a finitely generated ring.) We will prove that it maps the positive cone of $K_0(A(\mathcal{C},\pi))$ onto that of $\mathrm{F}_{\mathcal{C}_1}[\sigma^{-1}]$. Denote by $\Lambda'$ the set of simple objects that occur in some tensor power of $\sigma$. Then the aforementioned stationarity refers to the fact that,
        since $1\prec \sigma$ (cf.\ the proof of Lemma \ref{lemma:complete} below), $K_0(A(\mathcal{C},\pi))$ is isomorphic (as an ordered group) to the limit of the inductive sequence
        $$
            G_1\overset{M_{\sigma}}{\longrightarrow} G_2\overset{M_{\sigma}}{\longrightarrow} G_3\overset{M_{\sigma}}{\longrightarrow} \cdots,
        $$
        where $G_n = \mathbb{Z}^{(\Lambda')} = 
        \bigoplus_{\mu\in\Lambda'}\mathbb{Z}\mu$ for all $n$ and $M_{\sigma}$ 
        is viewed as a matrix in $M_{\Lambda'}(\mathbb{Z})$. Consider now an 
        arbitrary element of the limit of this inductive sequence. It can be 
        viewed as the equivalence class $[v,n]$ of an element $v\in G_n$ (for 
        some $n$). It is well-known that, in this stationary situation, $[v,n]$ 
        is a non-zero positive element in the limit if and only if $\langle 
        v,w\rangle > 0$, where $w$ is the Perron--Frobenius eigenvector of 
        $M_\sigma$, which is equal to $\sum_{\mu\in\Lambda'}d(\mu)\mu$ by 
        Remark \ref{remark:PF} and Frobenius Reciprocity. Thus, $[v,n]>0$ if 
        and only if $d(\phi([v,n])) > 0$. The claim follows.
    \end{remark}

    \section{Explicit computations}\label{section:explicit_computations}

    We will next, in a variety of cases, explicitly describe the ring $K_0(A(\mathcal{C},\pi))$ in terms of generators and relations, i.e., as a quotient
    $$
        \mathbb{Z}[t_1,\ldots,t_r]/\langle P_1(t_1,\ldots,t_r),\ldots,P_m(t_1,\ldots,t_r)\rangle.
    $$
    The computations for SU(3), Sp(4) and G$_2$ have, as far as we are aware, 
    not appeared in the literature before, except that the computation for 
    $\mathbf{Rep}(\mathrm{SU}(3))$ essentially appeared in the work of 
    Handelman--Rossmann \cite{HR1,HR2} and Handelman \cite{H1,H2}. The 
    SU(2)-computations were done by Wassermann \cite{Wa} in the case of 
    $\mathbf{Rep}(\mathrm{SU}(2))$ and by Evans--Gould \cite{EG} in the case of 
    $\mathbf{Rep}_k(\mathrm{SU}(2))$. The only novelty in these cases is a 
    slight clarification of the ring structure.

    \begin{example}
    As a warm-up, we give an explicit computation for the irreducible 
    $2$-dimensional representation in $\mathcal{C} = \mathbf{Rep}(S_3)$, where 
    $S_3$ is the symmetric group on three letters.
In this case $\Lambda = \{1,s,\pi\}$, where $1$ is the trivial representation, $s$ is the sign representation, and $\pi$ is the $2$-dimensional representation $\pi\colon S_3\to \mathcal{U}(\mathbb{C}^3\ominus \{(1,1,1)\})$ that permutes the coordinates. They satisfy the fusion rules $s\otimes s \cong 1$, $s\otimes \pi \cong \pi$ and $\pi\otimes\pi \cong 1\oplus s\oplus\pi$ (and $1$ is the tensor unit).

    We claim that
    $$
        K_0(A(\mathcal{C},\pi))\cong\mathbb{Z}[t]/\langle 1-t-2t^2\rangle
    $$
    as ordered rings, where the positive cone on the right hand side is $\{[p(t)]\,:\,p(1/2)>0\}\cup\{[0]\}$ (and $[p(t)]$ denotes the coset of $p(t)\in\mathbb{Z}[t]$). Note that, as we saw in Example \ref{example:product_type}, $K_0(A(\mathcal{C},\pi)) \cong K_0(M_{2^\infty}^{S_3})$ and that, since $1/2$ is a root of $1-t-2t^2$, it makes sense to evaluate a coset in $\mathbb{Z}[t]/\langle 1-t-2t^2\rangle$ at $1/2$.

    Since $\pi$ is self-conjugate and every simple summand occurs in $\pi^{\otimes 2}$, the map $\phi\colon K_0(A(\mathcal{C},\pi))\to F_{\mathcal{C}}[\pi^{-1}]$ from section \ref{subsection:ring_structure} is an isomorphism of ordered rings, where the positive cone on the right hand side is $\{x\,:\,d(x)>0\}\cup\{0\}$. Here, the fusion ring $F_{\mathcal{C}}$ is the representation ring $R(S_3)$ of $S_3$ and the quantum dimension $d$ in $\mathcal{C}$ is just the vector space dimension. (Note that $F_{\mathcal{C}}[\pi^{-1}] = F_{\mathcal{C}}[(\bar{\pi}\pi)^{-1}]$.)

    Define a map $\Phi\colon R(S_3)\to \mathbb{Z}[t]/\langle 1-t-2t^2\rangle$ by $\Phi(1) = \Phi(s) = [1]$ and $\Phi(\pi) = [2t+1]$. This is a ring homomorphism, since $\Phi(s^2) = \Phi(s)^2 = [1]$, $\Phi(s\cdot\pi) = \Phi(s)\cdot\Phi(\pi) = \Phi(\pi)$ and
    \begin{gather*}
        \Phi(\pi^2) = \Phi(1+s+\pi) = [2t+3] = [2t+1]^2 = \Phi(\pi)^2.
    \end{gather*}
    As $\Phi(\pi) = [2t+1]$ is an invertible element in $\mathbb{Z}[t]/\langle 1-t-2t^2\rangle$ (with $[2t+1]^{-1} = [t]$), we get an induced ring homomorphism $R(S_3)[\pi^{-1}]\to \mathbb{Z}[t]/\langle 1-t-2t^2\rangle$ (also denoted $\Phi$), which is surjective because $\Phi(1) = [1]$ and $\Phi(1/\pi) = [t]$. It is easy to verify that, when $\mathbb{Z}[t]/\langle 1-t-2t^2\rangle$ is equipped with the aforementioned positive cone, $\Phi$ is in fact an isomorphism of ordered rings. The claim follows.\end{example}

    \subsection{SU(2)$_k$}

    Evans and Gould proved the following result in \cite{EG}.

    \begin{theorem}[Evans--Gould]\label{theorem:SU2k} As ordered groups,
    $$
        K_0(A(\mathbf{Rep}_k(\mathrm{SU}(2)),\pi_1))\cong \mathbb{Z}[t]/I_k,
    $$
    where $I_k = \langle P_{k+1}^{\mathrm{SU}(2)}(t)\rangle$.

    Here, the polynomial $P_\lambda^{\mathrm{SU}(2)}(t)$ (for $\lambda\in\mathbb{N}_0$) is obtained from the Laurent polynomial $x^{-\lambda}Q_\lambda^{\mathrm{SU}(2)}(x)$ by performing the change of variables $t = 1/x^{2}$, and the positive cone on the right hand side is
    $$
        \big(\mathbb{Z}[t]/I_k\big)_+ = \{[p(t)]\,:\,p(\alpha_k)>0\}\cup\{[0]\},
    $$
    where $\alpha_k = d(\pi_1)^{-2} = [4\cos^2(\pi/(k+2))]^{-1}$.
    \end{theorem}

    \noindent Note that the definition of the positive cone makes sense, since $\alpha_k$ is a common zero of the polynomials in $I_k$ by the discussion in section \ref{subsection:fusion_rings}.

    \begin{remark}
        The polynomials $P_\lambda^{\mathrm{SU}(2)}(t)$ first appeared in \cite{J}, where their positivity properties were used to prove Jones' Index Rigidity Theorem.
    \end{remark}

    \begin{remark}
    It is easy to prove that the above identification is in fact one of ordered 
    rings. The proof employs similar ideas to the ones used in section 
    \ref{subsection:SU3k}, but is quite a bit simpler. Also, as mentioned in 
    section \ref{subsubsection:TL}, $A(\mathbf{Rep}_k(\mathrm{SU}(2)),\pi_1)$ 
    is an inductive limit of Temperley--Lieb--Jones algebras, and the 
    $*$-homomorphism $\theta$ that induces the product in $K$-theory has a 
    particularly nice diagrammatic description in terms of superimposed 
    Temperley--Lieb diagrams.
    \end{remark}

    \subsection{SU(2)}

    The preceding computation is closely related to the following earlier result from Wassermann's thesis \cite{Wa}.

    \begin{theorem}[Wassermann]\label{theorem:SU2} As ordered rings,
    $$
        K_0(A(\mathbf{Rep}(\mathrm{SU}(2)),\pi_1))\cong \mathbb{Z}[t],
    $$
    where the positive cone on the right hand side is
    $$
        \mathbb{Z}[t]_+ = \{p(t)\,:\,p>0\text{ on }(0,1/4]\}\cup\{0\}
    $$
    and the product on $K_0(A(\mathbf{Rep}(\mathrm{SU}(2)),\pi_1))\cong K_0(M_{2^\infty}^{\mathrm{SU}(2)})$ is induced by the $*$-homomorphism $M_{2^\infty}\otimes M_{2^\infty}\to M_{2^\infty}$ that interlaces the tensor factors.
    \end{theorem}

    \begin{remark}
    By the same work as that cited in section \ref{subsubsection:TL}, 
    $A(\mathbf{Rep}(\mathrm{SU}(2)),\pi_1)$ is also an inductive limit of 
    Temperley--Lieb--Jones algebras (with parameter $\delta = 2$), and the 
    $*$-homomorphism $\theta$ has the same description in terms of superimposed 
    Temperley--Lieb diagrams.
    \end{remark}

    \subsection{SU(3)$_k$}\label{subsection:SU3k}

    We will next prove the following result.

    \begin{theorem}\label{theorem:SU3k}
    As ordered rings,
    $$
        K_0(A(\mathbf{Rep}_k(\mathrm{SU}(3)),\pi_{(1,0)}))\cong \mathbb{Z}[s,t]/I_k,
    $$
    where $I_k = \langle P_{(k+1,0)}^{\mathrm{SU}(3)}(s,t),P_{(k+2,0)}^{\mathrm{SU}(3)}(s,t)\rangle$.

    Here, the polynomial $P_{\vec\lambda}^{\mathrm{SU}(3)}(s,t)$ (for
    $\vec\lambda=(\lambda_1,\lambda_2)\in \mathbb{N}_0^{\times 2}$) is obtained from the Laurent polynomial
    ${x}^{-\lambda_1}{y}^{-\lambda_2}Q_{\vec\lambda}^{\mathrm{SU}(3)}(x,y)$
    by performing the change of variables $(s,t) = (x/y^{2},y/x^{2})$, and the
    positive cone on the right hand side is
    $$
        \big(\mathbb{Z}[s,t]/I_k\big)_+ = \{[p(s,t)]\,:\,p(\beta_k,\beta_k)>0\}\cup\{[0]\},
    $$
    where $\beta_k = d(\pi_{(1,0)})^{-1} = [1+2\cos(2\pi/(k+3))]^{-1}$.
    \end{theorem}

    \noindent Again the definition of the positive cone makes sense because $(\beta_k,\beta_k)$ is a common zero of the polynomials in $I_k$ by the discussion in section \ref{subsection:fusion_rings}.

    \begin{remark}\label{remark:homom_A2}
        One can show that the AF-algebra 
        $A(\mathbf{Rep}_k(\mathrm{SU}(3)),\pi_{(1,0)})$ is an inductive limit 
        of so-called A$_2$-Temperley--Lieb--Jones algebras (cf.\ the references 
        in Remark \ref{remark:history}) and that the $*$-homomorphism $\theta$ 
        has a description in terms of superimposed A$_2$-Temperley--Lieb 
        diagrams.
    \end{remark}

    As the proof of Theorem \ref{theorem:SU3k} is rather lengthy, we will spend a moment describing the overall strategy. We will first define a group homomorphism $$\psi\colon K_0(A(\mathbf{Rep}_k(\mathrm{SU}(3)),\pi_{(1,0)}))\to \mathbb{Z}[s,t]/I_k$$
    and show that it is an isomorphism of ordered groups. Next, we will show that we have a commutative diagram
    $$
        \begin{CD}
            K_0(A(\mathbf{Rep}_k(\mathrm{SU}(3)),\pi_{(1,0)})) @>\psi >> \mathbb{Z}[s,t]/I_k\\
            @V\phi VV @V\psi_{\mathrm{cv}} VV\\
            \mathrm{Ver}_k(\mathrm{SU}(3))[\sigma^{-1}] @>\psi_{\mathrm{GF}}>> \mathbb{Z}[x^{\pm 1},y^{\pm 1}]/\tilde{J}_k
        \end{CD}
    $$
    where $\phi$ is the injective ring homomorphism from section \ref{subsection:ring_structure}, $\psi_{\mathrm{cv}}$ is given by the aforementioned change of variables, $\tilde{J}_k$ is the image of the fusion ideal $J_k(\mathrm{SU}(3))\subset\mathbb{Z}[x,y]$ in the localization $(\mathbb{Z}[x,y])[(xy)^{-1}] = \mathbb{Z}[x^{\pm 1},y^{\pm 1}]$, and $\psi_{\mathrm{GF}}$ is induced by the ring isomorphism from Theorem \ref{theorem:Gepner}. Finally, we conclude that $\psi = \psi_{\mathrm{cv}}^{-1}\circ\psi_{\mathrm{GF}}\circ\phi$ is in fact an isomorphism of ordered rings.

    \begin{remark}
One can deduce from the above argument that $I_k =
        \{p(s,t)\in\mathbb{Z}[s,t]\,:\,p(x/y^2,y/x^2) = 0\text{ for all
        }(x,y)\in V\setminus\{(0,0)\}\}$, where $V$ is the fusion variety
        associated to $\mathbf{Rep}_k(\mathrm{SU}(3))$, cf.\ section \ref{subsection:fusion_rings}. Moreover, if we
        \emph{define} $I_k$ by this formula then we can right away
        \emph{define} $\psi$ as the composition
        $\psi_{\mathrm{cv}}^{-1}\circ\psi_{\mathrm{GF}}\circ\phi$, and then the
        surjectivity of $\psi$, which takes up most of the following proof, is
        essentially automatic, since we know from the outset that $\psi$ is
        a ring homomorphism. However, in this approach it is not clear how to
        prove that $I_k$ has the stated generating set.
    \end{remark}

    Before we move on to the proof of the theorem, we must go over some preliminaries. We will use the notation $|\vec\lambda| = \lambda_1+\lambda_2$ for $\vec\lambda = (\lambda_1,\lambda_2)\in\mathbb{N}_0^{\times 2}$. Then $|\vec\lambda|$ is precisely the ``level'' of $\vec\lambda$ that we denoted by $\ell(\vec\lambda)$ in section \ref{subsection:fusion_rings}.

\subsubsection{On the polynomials $Q_{\vec\lambda}^{\mathrm{SU}(3)}(x,y)$ and $P_{\vec\lambda}^{\mathrm{SU}(3)}(s,t)$}\label{subsubsection:polys}

    By the well-known fusion rules of $\mathbf{Rep}(\mathrm{SU}(3))$, cf.\ Example \ref{example:SU3_fusion_rules}, the polynomials
    $Q_{\vec\lambda}^{\mathrm{SU}(3)}(x,y)$ are uniquely determined by the
    conditions $Q^{\mathrm{SU}(3)}_{\vec0}(x,y) = 1$,
    $Q^{\mathrm{SU}(3)}_{(1,0)}(x,y) = x$,
$$
    xQ^{\mathrm{SU}(3)}_{\vec\lambda}(x,y) = Q^{\mathrm{SU}(3)}_{\vec\lambda+(1,0)}(x,y)+Q^{\mathrm{SU}(3)}_{\vec\lambda+(0,-1)}(x,y)+Q^{\mathrm{SU}(3)}_{\vec\lambda+(-1,1)}(x,y)
$$
(with the convention that $Q^{\mathrm{SU}(3)}_{\vec\lambda}(x,y) = 0$ if $\vec\lambda\notin\mathbb{N}_0^{\times 2}$), and $$Q^{\mathrm{SU}(3)}_{(\lambda_1,\lambda_2)}(y,x) = Q^{\mathrm{SU}(3)}_{(\lambda_2,\lambda_1)}(x,y).$$ The first few of these polynomials are given in Table \ref{tab:1}.

\begin{table}
\caption{The polynomials $Q^{\mathrm{SU}(3)}_{\vec\lambda}(x,y)$ with $\lambda_1\geq\lambda_2$ and $|\vec\lambda|\leq 6$}
\label{tab:1}
\centering
\begin{tabular}{cl}
\hline\noalign{\smallskip}
$\vec\lambda$ & $Q^{\mathrm{SU}(3)}_{\vec\lambda}(x,y)$  \\
\noalign{\smallskip}\hline\noalign{\smallskip}
$(0,0)$ & 1 \\
\noalign{\smallskip}
$(1,0)$ & $x$ \\
\noalign{\smallskip}
$(2,0)$ & $x^2-y$ \\
$(1,1)$ & $xy-1$ \\
\noalign{\smallskip}
$(3,0)$ & $x^3-2xy+1$ \\
$(2,1)$ & $x^2y-y^2-x$ \\
\noalign{\smallskip}
$(4,0)$ & $x^4-3x^2y+y^2+2x$ \\
$(3,1)$ & $x^3y-2xy^2-x^2+2y$ \\
$(2,2)$ & $x^2y^2-x^3-y^3$ \\
\noalign{\smallskip}
$(5,0)$ & $x^5-4x^3y+3xy^2+3x^2-2y$ \\
$(4,1)$ & $x^4y-3x^2y^2+y^3-x^3+4xy-1$ \\
$(3,2)$ & $x^3y^2-2xy^3-x^4+x^2y+2y^2-x$ \\
\noalign{\smallskip}
$(6,0)$ & $x^6-5x^4y+6x^2y^2-y^3+4x^3-6xy+1$ \\
$(5,1)$ & $x^5y-4x^3y^2-x^4+3xy^3+6x^2y-3y^2-2x$ \\
$(4,2)$ & $x^4y^2-3x^2y^3-x^5+y^4+2x^3y+3xy^2-2x^2-y$ \\
$(3,3)$ & $x^3y^3-2x^4y-2xy^4+3x^2y^2+2x^3+2y^3-5xy+1$ \\
\noalign{\smallskip}\hline
\end{tabular}
\end{table}

We will next show that the expression $x^{-\lambda_1}y^{-\lambda_2}Q^{\mathrm{SU}(3)}_{\vec\lambda}(x,y)$ is a polynomial in $s = x/y^2$ and $t = y/x^2$. This follows immediately from the next two lemmas.

\begin{lemma}\label{lem:6} Let $\vec{\lambda}\in \mathbb{N}_0^{\times 2}$ be given. Then there exist integers $\gamma_{\vec{\nu}}$ such that
$$
    Q^{\mathrm{SU}(3)}_{\vec{\lambda}}(x,y) = x^{\lambda_1}y^{\lambda_2}+\sum_{\vec\nu\in\mathbb{N}_0^{\times 2}}\gamma_{\vec{\nu}}x^{\nu_1}y^{\nu_2},
$$
where $\gamma_{\vec\nu} = 0$ unless the following conditions are satisfied: $\nu_1+2\nu_2\leq \lambda_1+2\lambda_2$, $2\nu_1+\nu_2\leq 2\lambda_1+\lambda_2$, $|\vec\nu|<|\vec\lambda|$ and $\nu_1-\nu_2\equiv\lambda_1-\lambda_2$ (mod 3).
\end{lemma}

\begin{proof}
This follows by induction on $|\vec\lambda|$.
\end{proof}

\begin{lemma}\label{lem:6b}
Let $a,b\in\mathbb{Z}$ be given. Then $x^ay^b$ is a polynomial in $s = x/y^2$ and $t = y/x^2$ if and only if $a\equiv b$ (mod 3), $a+2b\leq 0$ and $2a+b\leq 0$.
\end{lemma}

\begin{proof}
Note that $x^ay^b = s^k t^l$ with $k = -(a+2b)/3$ and $l = -(2a+b)/3$.
\end{proof}

\noindent Denote by $P^{\mathrm{SU}(3)}_{\vec\lambda}(s,t)$ the polynomial obtained by performing the change of variables $s = x/y^2$ and $t = y/x^2$ in the expression $x^{-\lambda_1}y^{-\lambda_2}Q^{\mathrm{SU}(3)}_{\vec\lambda}(x,y)$. The first few of these polynomials are given in Table \ref{tab:2}. Note that $P^{\mathrm{SU}(3)}_{(\lambda_1,\lambda_2)}(t,s) = P^{\mathrm{SU}(3)}_{(\lambda_2,\lambda_1)}(s,t)$ for all $\vec\lambda\in\mathbb{N}_0^{\times 2}$, and that\label{page:recursion} we have the recursion formulae
$$
    P^{\mathrm{SU}(3)}_{\vec{\lambda}}(s,t) = \begin{cases} P^{\mathrm{SU}(3)}_{\vec{\lambda}+(-1,0)}(s,t)-stP^{\mathrm{SU}(3)}_{\vec{\lambda}+(-1,-1)}(s,t)-tP^{\mathrm{SU}(3)}_{\vec{\lambda}+(-2,1)}(s,t)&\!\!\!\text{if }\lambda_1\geq 1,\\
    P^{\mathrm{SU}(3)}_{\vec{\lambda}+(0,-1)}(s,t)-stP^{\mathrm{SU}(3)}_{\vec{\lambda}+(-1,-1)}(s,t)-sP^{\mathrm{SU}(3)}_{\vec{\lambda}+(1,-2)}(s,t)&\!\!\!\text{if }\lambda_2\geq 1
    \end{cases}
$$
(with the convention that $P^{\mathrm{SU}(3)}_{\vec{\lambda}}(s,t) = 0$ if $\vec\lambda\notin\mathbb{N}_0^{\times 2}$). In particular,
$$
    P^{\mathrm{SU}(3)}_{\vec{\lambda}}(0,t) = \begin{cases}P^{\mathrm{SU}(3)}_{\vec{\lambda}+(-1,0)}(0,t)-tP^{\mathrm{SU}(3)}_{\vec{\lambda}+(-2,1)}(0,t) &\text{if }\lambda_1\geq 1,\\
    P^{\mathrm{SU}(3)}_{\vec{\lambda}+(0,-1)}(0,t)&\text{if }\lambda_2\geq 1.\end{cases}
$$
We deduce from these identities that $P^{\mathrm{SU}(3)}_{\vec\lambda}(0,t)$ is independent of $\lambda_2$ and (by using that $P_0^{\mathrm{SU}(2)}(t) = P_1^{\mathrm{SU}(2)}(t)=1$ and $P_{\lambda+1}^{\mathrm{SU}(2)} = P_\lambda^{\mathrm{SU}(2)}(t)-tP_{\lambda-1}^{\mathrm{SU}(2)}(t)$ for $\lambda\in\mathbb{N}$) that we in fact have the following lemma.

\begin{lemma}\label{lem:N}
    For every $\vec\lambda\in\mathbb{N}_0^{\times 2}$, we have the identities $P^{\mathrm{SU}(3)}_{\vec\lambda}(0,t) = P^{\mathrm{SU}(2)}_{\lambda_1}(t)$ and $P^{\mathrm{SU}(3)}_{\vec\lambda}(s,0) = P^{\mathrm{SU}(2)}_{\lambda_2}(s)$.
\end{lemma}

\begin{table}
\caption{The polynomials $P^{\mathrm{SU}(3)}_{\vec\lambda}(s,t)$ with $\lambda_1\geq\lambda_2$ and $|\vec\lambda|\leq 6$}
\label{tab:2}
\centering
\begin{tabular}{cl}
\hline\noalign{\smallskip}
$\vec\lambda$ & $P^{\mathrm{SU}(3)}_{\vec\lambda}(s,t)$  \\
\noalign{\smallskip}\hline\noalign{\smallskip}
$(0,0)$ & 1 \\
\noalign{\smallskip}
$(1,0)$ & $1$ \\
\noalign{\smallskip}
$(2,0)$ & $1-t$ \\
$(1,1)$ & $1-st$ \\
\noalign{\smallskip}
$(3,0)$ & $1-2t+st^2$ \\
$(2,1)$ & $1-t-st$ \\
\noalign{\smallskip}
$(4,0)$ & $1-3t+t^2+2st^2$ \\
$(3,1)$ & $1-2t-st+2st^2$ \\
$(2,2)$ & $1-s-t$ \\
\noalign{\smallskip}
$(5,0)$ & $1-4t+3t^2+3st^2-2st^3$ \\
$(4,1)$ & $1-3t+t^2-st+4st^2-s^2t^3$ \\
$(3,2)$ & $1-2t-s+st+2st^2-s^2t^2$ \\
\noalign{\smallskip}
$(6,0)$ & $1-5t+6t^2-t^3+4st^2-6st^3+s^2t^4$ \\
$(5,1)$ & $1-4t-st+3t^2+6st^2-3st^3-2s^2t^3$ \\
$(4,2)$ & $1-3t-s+t^2+2st+3st^2-2s^2t^2-s^2t^3$ \\
$(3,3)$ & $1-2s-2t+3st+2s^2t+2st^2-5s^2t^2+s^3t^3$ \\
\noalign{\smallskip}\hline
\end{tabular}
\end{table}

\noindent The next lemma is analogous to the following special case of the results in \cite{G}: $$J_k(\mathrm{SU}(3)) = \langle Q_{(k+1,0)}^{\mathrm{SU}(3)}(x,y),Q_{(k+2,0)}^{\mathrm{SU}(3)}(x,y) \rangle = \langle Q_{\vec\lambda}^{\mathrm{SU}(3)}(x,y)\,:\,|\vec\lambda| = k+1\rangle.$$

\begin{lemma}\label{lemma:generation}
    For every $k\in\mathbb{N}$,
    $$
        I_k \df \langle P_{(k+1,0)}^{\mathrm{SU}(3)}(s,t),P_{(k+2,0)}^{\mathrm{SU}(3)}(s,t) \rangle = \langle P_{\vec\lambda}^{\mathrm{SU}(3)}(s,t)\,:\,|\vec\lambda| = k+1\rangle.
    $$
\end{lemma}

\begin{proof}
    Fix $k\in\mathbb{N}$. Put $I = \langle P_{\vec\lambda}^{\mathrm{SU}(3)}(s,t)\,:\,|\vec\lambda| = k+1\rangle$. Since
    $
        P^{\mathrm{SU}(3)}_{(k+2,0)}(s,t) = P^{\mathrm{SU}(3)}_{(k+1,0)}(s,t)-tP^{\mathrm{SU}(3)}_{(k,1)}(s,t)\in I,
    $
    we have that $I_k\subset I$.

    Put $I^{(j)}=\langle\{P^{\mathrm{SU}(3)}_{(i,k+1-i)}(s,t)\,:\,i=j+1,\ldots,k+1\}\cup\{P^{\mathrm{SU}(3)}_{(k+2,0)}(s,t)\}\rangle$ for $j=0,\ldots,k$ (so that $I^{(0)}\supset I^{(1)}\supset\cdots\supset I^{(k)} = I_k$). We claim that
    $$
        P^{\mathrm{SU}(3)}_{(j,k+1-j)}(s,t)\in I^{(j)}
    $$
    for all $j = 0,\ldots,k$. The claim implies first that $P^{\mathrm{SU}(3)}_{(k,1)}(s,t)\in I^{(k)}$, whereby $I^{(k-1)} = I^{(k)}$, next that $P^{\mathrm{SU}(3)}_{(k-1,2)}(s,t)\in I^{(k-1)}$, whereby $I^{(k-2)} = I^{(k)}$, etc. Continuing in this way, the claim implies that $I\subset I^{(k)} = I_k$.

    We skip the easy proof of the claim, which proceeds by induction on $j$ using the identity
$P^{\mathrm{SU}(3)}_{(0,k+1)}(s,t) = P^{\mathrm{SU}(3)}_{(1,k)}(s,t)-sP^{\mathrm{SU}(3)}_{(2,k-1)}(s,t)$, which implies that $P^{\mathrm{SU}(3)}_{(0,k+1)}(s,t)\in I^{(0)}$, and the identity
    $$
        tP^{\mathrm{SU}(3)}_{(j,k+1-j)}(s,t) = P^{\mathrm{SU}(3)}_{(j+1,k-j)}(s,t)-P^{\mathrm{SU}(3)}_{(j+2,k-j-1)}(s,t)+sP^{\mathrm{SU}(3)}_{(j+3,k-j-2)}(s,t),
    $$
    which is valid for $j = 0,\ldots,k-1$.
\end{proof}

In the statement of the following lemma, we use the floor function $\lfloor\cdot\rfloor$ defined by $\lfloor x\rfloor = \max\{n\in\mathbb{Z}\,:\,n\leq x\}$ for $x\in\mathbb{R}$.

\begin{lemma}\label{lem:9}
Given $\vec{\lambda}\in\mathbb{N}_0^{\times 2}$, we can write
\begin{equation}
    P^{\mathrm{SU}(3)}_{\vec{\lambda}}(s,t) = \sum_{j=0}^{n} (st)^j\big(A_j(s)+B_j(t)\big)
\end{equation}
for some $n\in\mathbb{N}_0$, $A_j(s)\in\mathbb{Z}[s]$ with $\deg A_j(s)\leq \lfloor \lambda_2/2\rfloor$, and $B_j(t)\in\mathbb{Z}[t]$ with $\deg B_j(t)\leq \lfloor \lambda_1/2\rfloor$. Moreover, one of the following two additional sets of conditions (but not both) may be arranged.
\begin{itemize}
    \item[(i)] If $\lambda_2$ is even then we may arrange that $\deg A_j(s)\leq
    \lfloor\lambda_2/2\rfloor-1$ for $j>0$ and that $B_j(0)=0$ for all $j$.
    \item[(ii)] If $\lambda_1$ is even then we may arrange that $\deg B_j(t)\leq
    \lfloor\lambda_1/2\rfloor-1$ for $j>0$ and that $A_j(0)=0$ for all $j$.
\end{itemize}
\end{lemma}

\begin{proof}
We proceed by induction on $|\vec\lambda|$. The basis for the induction follows by inspection of Table~2. Let next $\vec{\lambda}\in\mathbb{N}_0^{\times 2}$ be given and assume that the assertion is true at all ``levels'' $<|\vec\lambda|$.

Suppose first that $\lambda_2\geq 1$. Then, by induction hypothesis, we can write
\begin{align*}
    \notag P^{\mathrm{SU}(3)}_{\vec{\lambda}+(0,-1)}(s,t) &= \sum_{j\geq 0} (st)^j\big(C_j(s)+D_j(t)\big),\\
    \notag P^{\mathrm{SU}(3)}_{\vec{\lambda}+(-1,-1)}(s,t) &= \sum_{j\geq 0} (st)^j\big(E_j(s)+F_j(t)\big),\\
    P^{\mathrm{SU}(3)}_{\vec{\lambda}+(1,-2)}(s,t) &= \sum_{j\geq 0} (st)^j\big(G_j(s)+H_j(t)\big)
\end{align*}
with polynomials $C_j(s)$, $D_j(t)$, $\ldots$ that satisfy
\begin{align*}
    \notag &\deg C_j(s)\leq \left\lfloor\frac{\lambda_2-1}{2}\right\rfloor\!\!,\,
    \deg E_j(s)\leq \left\lfloor\frac{\lambda_2-1}{2}\right\rfloor\!\!,\,
    \deg G_j(s)\leq \left\lfloor\frac{\lambda_2-2}{2}\right\rfloor\!\!,\\
    &\qquad\qquad\deg D_j(t)\leq \left\lfloor\frac{\lambda_1}{2}\right\rfloor\!\!,\,
    \deg F_j(t)\leq \left\lfloor\frac{\lambda_1-1}{2}\right\rfloor\!\!,\,
    \deg H_j(t)\leq \left\lfloor\frac{\lambda_1+1}{2}\right\rfloor
\end{align*}
for all $j$. If $\lambda_1 = 0$ then $P^{\mathrm{SU}(3)}_{\vec{\lambda}+(-1,-1)}(s,t) = 0$, in which case we may assume that $E_j(s) = F_j(t) = 0$ for all $j$. On the other hand, if $\lambda_2 = 1$ then $P^{\mathrm{SU}(3)}_{\vec{\lambda}+(1,-2)}(s,t) = 0$, in which case we may assume that $G_j(s) = H_j(t) = 0$ for all $j$.

Setting $E_{-1}(s) = F_{-1}(t) = H_{-1}(t)=0$, put $A_j(s) = C_j(s)-E_{j-1}(s)-sG_j(s)-sH_j(0)\in\mathbb{Z}[s]$ and $B_j(t) = D_j(t)-F_{j-1}(t)-t^{-1}(H_{j-1}(t)-H_{j-1}(0))\in\mathbb{Z}[t]$ for all $j$. It follows from the recursion formulae on page \pageref{page:recursion} that
\begin{equation*}
    P^{\mathrm{SU}(3)}_{\vec{\lambda}}(s,t) = \sum_{j\geq 0} (st)^j\big(A_j(s)+B_j(t)\big).
\end{equation*}
Here, if $\lambda_2\geq 2$ then
\begin{equation*}
    \deg A_j(s)\leq \max\left\{
    \left\lfloor\frac{\lambda_2-1}{2}\right\rfloor,
    \left\lfloor\frac{\lambda_2-2}{2}\right\rfloor+1,1\right\}
    =\left\lfloor\frac{\lambda_2}{2}\right\rfloor
\end{equation*}
for all $j$, while if $\lambda_2 = 1$ then
\begin{equation*}
	\deg A_j(s)\leq \left\lfloor\frac{\lambda_2-1}{2}\right\rfloor
	= 0 = \left\lfloor \frac{\lambda_2}{2}\right\rfloor
\end{equation*}
for all $j$. Also,
\begin{equation*}
    \deg B_j(t)\leq
    \max\left\{\left\lfloor\frac{\lambda_1}{2}\right\rfloor,
    \left\lfloor\frac{\lambda_1+1}{2}\right\rfloor-1\right\}=\left\lfloor\frac{\lambda_1}{2}\right\rfloor
\end{equation*}
for all $j$.

Assume now that $\lambda_2$ is even. Then $\lambda_2-2$ is even and we may arrange that $\deg G_j(s)\leq \lfloor(\lambda_2-2)/2\rfloor-1$ for $j>0$ and $H_j(0) = 0$ for all $j$. Thus,
\begin{equation*}
    \deg A_j(s)\leq \max\left\{\left\lfloor\frac{\lambda_2-1}{2}\right\rfloor,
    \left\lfloor\frac{\lambda_2-2}{2}\right\rfloor\right\} =
    \left\lfloor\frac{\lambda_2}{2}\right\rfloor-1
\end{equation*}
for $j>0$. As $\left\lfloor\lambda_2/2\right\rfloor-1\geq 0$, we may replace $B_j(t)$ with $B_j(t)-B_j(0)$ and $A_j(s)$ with $A_j(s)+B_j(0)$ to ensure that $B_j(0)=0$ for all $j$ without altering the estimates on the degrees.

Assume next that $\lambda_1$ is even. Then we may arrange that $\deg D_j(t)\leq \lfloor\lambda_1/2\rfloor-1$ for $j>0$ and $C_j(0) = 0$ for all $j$. Thus,
\begin{equation*}
    \deg B_j(t) \leq
    \max\left\{\left\lfloor\frac{\lambda_1}{2}\right\rfloor-1,\left\lfloor\frac{\lambda_1-1}{2}\right\rfloor,\left\lfloor\frac{\lambda_1+1}{2}\right\rfloor-1\right\}
     = \left\lfloor\frac{\lambda_1}{2}\right\rfloor-1
\end{equation*}
for $j>0$. If $\lambda_1 = 0$ then $C_j(0)=E_{j-1}(0)=0$ so that $A_j(0)=0$ for all $j$. If $\lambda_1\geq 2$ then $\left\lfloor\lambda_1/2\right\rfloor-1\geq 0$ so that we may replace $B_j(t)$ with $B_j(t)+A_j(0)$ and $A_j(s)$ with $A_j(s)-A_j(0)$ to ensure that $A_j(0)=0$ for all $j$ without altering the estimates on the degrees.

If $\lambda_1\geq 1$ (but possibly $\lambda_2 = 0$) then one first applies the above arguments to $P_{(\lambda_2,\lambda_1)}^{\mathrm{SU}(3)}(s,t)$ and then interchanges $s$ and $t$.
\end{proof}

\subsubsection{On certain monomials $m_{\vec\lambda,n}(s,t)$}

We will now introduce two families of sets. Their relevance will become clear in the next section. For each $n\in\mathbb{N}_0$, we put
$$
	B_0(n) = \{\vec{\lambda}\in\mathbb{N}_0^{\times
	2}\,:\,\lambda_1\equiv\lambda_2\text{ (mod
	3)},\,\lambda_1+2\lambda_2\leq 3n,\,2\lambda_1+\lambda_2\leq 3n\}
$$
and
\begin{multline*}
	B_1(n) = \{\vec{\lambda}\in\mathbb{N}_0^{\times
	2}\,:\,\lambda_1-1\equiv \lambda_2\text{ (mod
	3)},\\ (\lambda_1-1)+2\lambda_2\leq 3n, 2(\lambda_1-1)+\lambda_2\leq
	3n\}.
\end{multline*}
Given $n\in\mathbb{N}_0$ and $\vec\lambda\in B_0(n)$, we denote by
$m_{\vec\lambda,n}(s,t)$ the monomial in $\mathbb{Z}[s,t]$ obtained by
	performing the change of variables
	$s = x/y^2$ and $t = y/x^2$ in the expression
	$x^{\lambda_1-n}y^{\lambda_2-n}$. The fact that this expression is indeed a
	monomial in $s$ and $t$ follows from Lemma \ref{lem:6b}.
Similarly, given $n\in\mathbb{N}_0$ and $\vec\lambda\in B_1(n)$, we denote
by $m_{\vec\lambda,n}(s,t)$ the monomial in
	$\mathbb{Z}[s,t]$ obtained by
	performing the same change of variables in the expression
	$x^{(\lambda_1-1)-n}y^{\lambda_2-n}$.
By definition, we then have the following identities, in which $s = x/y^2$ and
$t = y/x^2$.
\begin{equation}\label{equation:monomials}
    m_{\vec\lambda,n}(s,t)P_{\vec\lambda}^{\mathrm{SU}(3)}(s,t)
    = \begin{cases}Q_{\vec\lambda}^{\mathrm{SU}(3)}(x,y)/(x^ny^n)&\text{ if
    }\vec\lambda\in B_0(n),\\
    Q_{\vec\lambda}^{\mathrm{SU}(3)}(x,y)/(x^{n+1}y^n)&\text{ if
    }\vec\lambda\in B_1(n).\end{cases}
\end{equation}

\noindent Moreover, we have the following easy lemma.

\begin{lemma}\label{lemma:monomials2}
Let $n\in\mathbb{N}_0$ and $\vec{\lambda}\in B_j(n)$ be given, where $j\in
\{0,1\}$. Then there exist integers $\gamma_{\vec{\nu}}$ such that
\begin{itemize}
    \item[(i)] $x^{\lambda_1}y^{\lambda_2} = \sum_{\vec{\nu}\in \mathbb{N}_0^{\times 2}} \gamma_{\vec{\nu}}Q^{\mathrm{SU}(3)}_{\vec{\nu}}(x,y)$, and
    \item[(ii)] $m_{\vec\lambda,n}(s,t) = \sum_{\vec\nu\in \mathbb{N}_0^{\times 2}} \gamma_{\vec\nu}m_{\vec{\nu},n}(s,t)P^{\mathrm{SU}(3)}_{\vec\nu}(s,t)$,
\end{itemize}
where $\gamma_{\vec\nu} = 0$ unless $|\vec\nu|\leq |\vec\lambda|$ and
$\vec\nu\in B_j(n)$.
\end{lemma}

\begin{proof}
The statement (i) follows by induction on $|\vec\lambda|$ using Lemma \ref{lem:6}, whereafter (ii) follows from (i) and equation (\ref{equation:monomials}).
\end{proof}

\subsubsection{Definition, injectivity and positivity of $\psi$}\label{subsubsection:definition_psi}

Fix $k\in\mathbb{N}$. Recall the fusion rules of $\mathbf{Rep}_k(\mathrm{SU}(3))$ from Example \ref{example:SU3_fusion_rules}.

\begin{lemma}
	Consider the objects $\pi = \pi_{(1,0)}$, $\bar{\pi} = \pi_{(0,1)}$ and
	$\sigma = \pi\otimes\bar{\pi}$ in $\mathbf{Rep}_k(\mathrm{SU}(3))$. For
	every $n\in\mathbb{N}_0$, we have that
	\begin{itemize}
	\item[(i)] $\{\mu\in\Lambda\,:\,\mu\prec\sigma^{\otimes n}\} =
	\{\pi_{\vec\lambda}\,:\,|\vec\lambda|\leq k,\,\vec\lambda\in B_0(n)\}$;
	\item[(ii)] $\{\mu\in\Lambda\,:\,\mu\prec\sigma^{\otimes n}\otimes\pi\} =
	\{\pi_{\vec\lambda}\,:\,|\vec\lambda|\leq k,\,\vec\lambda\in B_1(n)\}$.
	\end{itemize}
\end{lemma}

\begin{proof}
	Due to Frobenius Reciprocity, the proof amounts to showing the following
	two statements.
\begin{itemize}
    \item[(A)] $		
		\{\pi_{\vec\lambda}\,:\,|\vec\lambda|\leq k,\,\vec\lambda\in B_1(n)\}
	$
	is the set of simple objects that occur as a direct summand of
	$\pi_{\vec\lambda}\otimes\pi$ for some $\vec\lambda\in B_0(n)$ with
	$|\vec\lambda|\leq k$.
    \item[(B)] $		
			\{\pi_{\vec\lambda}\,:\,|\vec\lambda|\leq k,\,\vec\lambda\in
			B_0(n+1)\}
		$
		is the set of simple objects that occur as a direct summand of
		$\pi_{\vec\lambda}\otimes\bar{\pi}$ for some $\vec\lambda\in B_1(n)$
		with
		$|\vec\lambda|\leq k$.
\end{itemize}
		We will only prove (A). The proof of (B) is similar. Let
		first $\vec\lambda\in B_1(n)$ with $|\vec\lambda|\leq k$ be given. If
		$\lambda_1\geq 1$ then
		$\pi_{\vec\lambda}$ occurs as a direct summand of
		$\pi_{\vec\lambda+(-1,0)}\otimes\pi$, where clearly $\vec\lambda+(-1,0)\in B_0(n)$, while if $\lambda_1 = 0$ then
		$\lambda_2\geq 2$, $\pi_{\vec\lambda}$ occurs as a
		direct summand of $\pi_{\vec\lambda+(1,-1)}\otimes\pi$, and it is easy to show that $\vec\lambda+(1,-1)\in B_0(n)$.
		
		Assume conversely that $\nu$ is a simple object occurring as a direct
		summand of $\pi_{\vec\lambda}\otimes\pi$ for some $\vec\lambda\in
		B_0(n)$ with
			$|\vec\lambda|\leq k$. Then
			$\nu = \pi_{\vec\mu}$ for some
			$\vec\mu\in\{\vec\lambda+(1,0),\vec\lambda+(0,-1),\vec\lambda+(-1,1)
			 \}\cap\mathbb{N}_0^{\times 2}\subset B_1(n)$.
\end{proof}

Let $\pi$, $\bar{\pi}$ and $\sigma$ be as in the preceding lemma. Then
$K_0(A(\mathbf{Rep}_k(\mathrm{SU}(3)),\pi))$ is isomorphic (as an ordered group)
to the limit of the inductive sequence
\begin{equation}\label{equation:system_SU3}
        \cdots \longrightarrow
        \bigoplus_{\mu\prec\sigma^{\otimes n}}\mathbb{Z}\mu
        \overset{M_\pi}{\longrightarrow}
        \bigoplus_{\mu\prec\sigma^{\otimes n}\otimes\pi}\mathbb{Z}\mu
        \overset{M_{\bar{\pi}}}{\longrightarrow}\bigoplus_{\mu\prec\sigma^{\otimes (n+1)}}\mathbb{Z}\mu\longrightarrow\cdots,
    \end{equation}
where $M_\pi$ (resp.\ $M_{\bar{\pi}}$) is defined by multiplication by $\pi$
(resp.\ $\bar{\pi}$) in the fusion
ring. Given $n\in\mathbb{N}_0$, we define
$$
    \psi_{0,n}\colon \bigoplus_{\mu\prec\sigma^{\otimes
    n}}\mathbb{Z}\mu\longrightarrow \mathbb{Z}[s,t]/I_k
$$
by $\psi_{0,n}(\pi_{\vec\lambda}) =
[m_{\vec\lambda,n}(s,t)P_{\vec\lambda}^{\mathrm{SU}(3)}(s,t)]$ for
$\vec\lambda\in B_0(n)$ with $|\vec\lambda|\leq k$ and
$$
    \psi_{1,n}\colon \bigoplus_{\mu\prec\sigma^{\otimes
    n}\otimes\pi}\mathbb{Z}\mu\longrightarrow \mathbb{Z}[s,t]/I_k
$$
by $\psi_{1,n}(\pi_{\vec\lambda}) =
[m_{\vec\lambda,n}(s,t)P_{\vec\lambda}^{\mathrm{SU}(3)}(s,t)]$ for
$\vec\lambda\in B_1(n)$ with $|\vec\lambda|\leq k$. To verify that these maps
induce a group homomorphism
$$
    \psi\colon
    K_0(A(\mathbf{Rep}_k(\mathrm{SU}(3)),\pi_{(1,0)}))\longrightarrow
    \mathbb{Z}[s,t]/I_k,
$$
we must show that, for each $n\in\mathbb{N}_0$,
\begin{itemize}
\item[(i)] $\psi_{0,n} = \psi_{1,n}\circ M_\pi$, and
\item[(ii)] $\psi_{1,n} = \psi_{0,n+1}\circ M_{\bar{\pi}}$.
\end{itemize}
We proceed to prove (i). Put $\Delta_{\vec\lambda,\vec\nu} =
(N_\pi)_{\pi_{\vec\lambda},\pi_{\vec\nu}}$ for
$\vec\lambda,\vec\nu\in\mathbb{N}_0^{\times 2}$ with
$|\vec\lambda|,|\vec\nu|\leq k$, where $N_\pi$ is the fusion matrix
of
$\pi$ in $\mathbf{Rep}_k(\mathrm{SU}(3))$. Let $\vec\lambda\in B_0(n)$ with
$|\vec\lambda|\leq k$ be given. Then
$$
	xQ_{\vec\lambda}^{\mathrm{SU}(3)}(x,y)-\sum_{|\vec\nu|\leq
	k}\Delta_{\vec\lambda,\vec\nu}Q_{\vec\nu}^{\mathrm{SU}(3)}(x,y)
	= \begin{cases}
		0&\text{ if }|\vec\lambda|<k,\\
		Q_{\vec\lambda+(1,0)}^{\mathrm{SU}(3)}(x,y)&\text{ if }|\vec\lambda| =
		k.
	\end{cases}
$$
By the statement (A) in the proof of the preceding lemma, we get that
$\Delta_{\vec\lambda,\vec\nu}\neq 0$ only if $\vec\nu\in B_1(n)$ (and
$|\vec\nu|\leq k$). Hence, we may divide by $x^{n+1}y^n$ and use equation
(\ref{equation:monomials}) on page \pageref{equation:monomials} to deduce that
\begin{multline*}
	m_{\vec\lambda,n}(s,t)P_{\vec\lambda}^{\mathrm{SU}(3)}(s,t)-\sum_{|\vec\nu|\leq
		k}\Delta_{\vec\lambda,\vec\nu}m_{\vec\nu,n}(s,t)P_{\vec\nu}^{\mathrm{SU}(3)}(s,t)\\
		= \begin{cases} 0&\text{ if }|\vec\lambda|<k,\\
		m_{\vec\lambda+(1,0),n}(s,t)P_{\vec\lambda+(1,0)}^{\mathrm{SU}(3)}(s,t)&\text{ if }|\vec\lambda| =
		k.\end{cases}
\end{multline*}
Since the right hand side belongs to $I_k$ by Lemma \ref{lemma:generation}, we conclude that (i) holds. The
proof of (ii) is similar. Thus, we have an induced group homomorphism $\psi$ as
above.

For the sake of convenience, we will from now on use the notation $$B(k,n) =
\{\vec\lambda\in
B_0(n)\,:\,|\vec\lambda|\leq k\}.$$

\begin{lemma}
	The group homomorphism $\psi$ is injective and positive.
\end{lemma}

\begin{proof}
To show that $\psi$ is injective, it
suffices to prove the following: If
$\psi_{0,n}(v) = [0]$ then there exists $N\in\mathbb{N}_0$ for which
$M_\sigma^N(v) = 0$ in the fusion ring. Write $v = \sum_{\vec\lambda\in B(k,n)}
v_{\vec\lambda}\pi_{\vec\lambda}$ with $v_{\vec\lambda}\in\mathbb{Z}$ for all $\vec\lambda$. Then the assumption $\psi_{0,n}(v) = [0]$
means that
$$
	\sum_{\vec\lambda\in B(k,n)}
	v_{\vec\lambda}m_{\vec\lambda,n}(s,t)P_{\vec\lambda}^{\mathrm{SU}(3)}(s,t)
	=
	p_1(s,t)P_{(k+1,0)}^{\mathrm{SU}(3)}(s,t)+p_2(s,t)P_{(k+2,0)}^{\mathrm{SU}(3)}(s,t)
$$
for some $p_j(s,t)\in\mathbb{Z}[s,t]$ ($j = 1,2$). Performing the change of
variables $s = x/y^2$ and $t = y/x^2$ and multiplying by a sufficiently high
power of $xy$, we deduce that
$$
	(xy)^N\sum_{\vec\lambda\in B(k,n)}
	v_{\vec\lambda}Q_{\vec\lambda}^{\mathrm{SU}(3)}(x,y) =
	q_1(x,y)Q_{(k+1,0)}^{\mathrm{SU}(3)}(x,y)+q_2(x,y)
	Q_{(k+2,0)}^{\mathrm{SU}(3)}(x,y)
$$
for some $N\in\mathbb{N}_0$ and $q_j(x,y)\in\mathbb{Z}[x,y]$ ($j=1,2$). Since
the right hand side belongs to the fusion ideal $J_k(\mathrm{SU}(3))$, this
implies that $M_\sigma^N(v)
= 0$, as desired.

It follows
by the same reasoning as that applied in Remark \ref{remark:ordered_ring_iso}
that $\psi$ maps the positive cone
$K_0(A(\mathbf{Rep}_k(\mathrm{SU}(3)),\pi_{(1,0)}))_+$ into
$\big(\mathbb{Z}[s,t]/I_k\big)_+$.
\end{proof}

\subsubsection{Proof of surjectivity of $\psi$}

According to Lemma \ref{lemma:monomials2}, we have that
$$
    [m_{\vec\lambda,n}(s,t)] = \psi_{0,n}\left(\sum_{\vec\nu\in B(k,n)}\gamma_{\vec\nu}\pi_{\vec\nu}\right)\in\mathrm{im}\,\psi
$$
for $n\in\mathbb{N}_0$ and $\vec\lambda\in B(k,n)$. This has the following consequence.

\begin{lemma}\label{lem:8}
For each $r\in\{0,1,\ldots,\lfloor k/3\rfloor\}$, we have that $[s^r],[t^r]\in\im\psi$.
\end{lemma}

\begin{proof}
Note that $m_{(3r,0),2r}(s,t) = s^r$ and $m_{(0,3r),2r}(s,t) = t^r$ for $r\geq 0$.
\end{proof}

\noindent We will also need the following remark.

\begin{remark}\label{rem:5}
Let $n_1,n_2\in\mathbb{N}_0$ and $\vec\nu\in B(k,n_1)$ be given. Then
$\vec\nu\in B(k,n_1+n_2)$ and, by performing the change of variables $s =
x/y^2$ and $t = y/x^2$ (under which $(xy)^{-1} = st$) in
the identity $x^{\lambda_1-(n_1+n_2)}y^{\lambda_2-(n_1+n_2)} =
(xy)^{-n_2}x^{\lambda_1-n_1}y^{\lambda_2-n_1}$, we get that
$m_{\vec\nu,n_1+n_2}(s,t)= (st)^{n_2}m_{\vec\nu,n_1}(s,t)$. It follows that if
$[q(s,t)]\in\im\psi$ then
$[(st)^jq(s,t)]\in\im\psi$ for all $j\in\mathbb{N}_0$.
\end{remark}

We are now ready to prove the surjectivity of $\psi$. Let us first prove it for $k=5$. This will guide us to the proof in the general case. By Lemma \ref{lem:8}, (the cosets of) $1$, $s$ and $t$ belong to the image of $\psi$. By Remark \ref{rem:5}, so do monomials obtained from these by multiplying by a power of $st$. Also, Lemma \ref{lemma:generation} implies that, in $\mathbb{Z}[s,t]/I_5$, we have the identity
$$
    0 \equiv P^{\mathrm{SU}(3)}_{(2,4)}(s,t) = s^2-3s+1-t+st(2+3s)+(st)^2(-2-s).
$$
Thus, $s^2$ is also in the image of $\psi$. By multiplying the above identity by $s$ repeatedly, we get that $s^j\in\im\psi$ for all $j\geq 0$. Similarly, the identity
$$
    0 \equiv P^{\mathrm{SU}(3)}_{(4,2)}(s,t) = t^2-3t+1-s+st(2+3t)+(st)^2(-2-t)
$$
shows, by repeated multiplication by $t$, that $t^j\in\im\psi$ for all $j\geq 0$. By using Remark \ref{rem:5} again, we conclude that $\psi$ is surjective for $k=5$.

In general, for $k\geq 6$, Lemma \ref{lem:8} shows that $1,s,\ldots,s^r$ and $1,t,\ldots,t^r$ belong to $\im\psi$, where $r = \lfloor k/3\rfloor$, and Lemmas \ref{lemma:generation} and \ref{lem:9} yield an identity
$$
    0 \equiv P^{\mathrm{SU}(3)}_{(2r,k+1-2r)}(s,t) = \sum_{i=0}^{n} (st)^i[A_i(s)+B_i(t)],
$$
where $A_i(s)\in\mathbb{Z}[s]$ with $\deg A_i(s)\leq \lfloor(k+1-2r)/2\rfloor$ and $B_i(t)\in\mathbb{Z}[t]$ with $\deg B_i(t)\leq r$ for all $i$. Moreover, we may assume that $A_0(0)
= 0$ and that $\deg B_i(t)\leq r-1$ for $i>0$. Since $P^{\mathrm{SU}(3)}_{(2r,k+1-2r)}(0,t) = P^{\mathrm{SU}(2)}_{2r}(t)$ by Lemma \ref{lem:N}, it is easy to verify that $B_0(t) = P^{\mathrm{SU}(2)}_{2r}(t)$ is a polynomial of degree $r$ with leading coefficient $\pm 1$. Also, since
$$
    \left\lfloor \frac{k+1-2r}{2}\right\rfloor = \left\lfloor
    \frac{k+1}{2}\right\rfloor-\left\lfloor\frac{k}{3}\right\rfloor\leq
    \left\lfloor \frac{k}{3}\right\rfloor = r
$$
for $k\geq 6$ (as is shown e.g.\ by splitting up into six cases depending on the residue class of $k$ modulo 6), we get that $\deg A_i(s)\leq r$ for all $i$. Thus, we can show that every power of $t$ belongs to $\mathrm{im}\,\psi$ by multiplying the above identity repeatedly by $t$ (and using Remark \ref{rem:5}).
Since one can similarly show that every power of $s$ belongs to $\mathrm{im}\,\psi$, it follows that $\psi$ is surjective whenever $k\geq 6$.

Finally, one can deal with each $k\in\{1,2,3,4\}$ in a similar fashion to conclude that $\psi$ is surjective for every positive integer $k$.

\subsubsection{Conclusion of the proof of Theorem \ref{theorem:SU3k}}

Consider the diagram
$$
        \begin{CD}
            K_0(A(\mathbf{Rep}_k(\mathrm{SU}(3)),\pi_{(1,0)})) @>\psi >> \mathbb{Z}[s,t]/I_k\\
            @V\phi VV @V\psi_{\mathrm{cv}} VV\\
            \mathrm{Ver}_k(\mathrm{SU}(3))[\sigma^{-1}] @>\psi_{\mathrm{GF}}>> \mathbb{Z}[x^{\pm 1},y^{\pm 1}]/\tilde{J}_k
        \end{CD}
    $$
    where $\phi$ is the injective ring homomorphism from section \ref{subsection:ring_structure}, $\tilde{J}_k$ is the image of the fusion ideal $J_k(\mathrm{SU}(3))\subset\mathbb{Z}[x,y]$ in the localization $(\mathbb{Z}[x,y])[(xy)^{-1}] = \mathbb{Z}[x^{\pm 1},y^{\pm 1}]$, $\psi_{\mathrm{GF}}$ is induced by the ring isomorphism from Theorem \ref{theorem:Gepner}, and $\psi_{\mathrm{cv}}$ is defined by
    $
        \psi_{\mathrm{cv}}([p(s,t)]) = [p(x/y^2,y/x^2)]
    $
    for $p(s,t)\in\mathbb{Z}[s,t]$. (Note that the definition of $I_k$ and Gepner's description of $J_k(\mathrm{SU}(3))$ immediately imply that $\psi_{\mathrm{cv}}$ is well-defined.) By equation (\ref{equation:monomials}) on page \pageref{equation:monomials}, this diagram commutes.

    Since $\psi$ is bijective and as $\phi$ and $\psi_{\mathrm{GF}}$ are injective, it follows that $\psi_{\mathrm{cv}}$ is injective as well. Since $\psi_{\mathrm{cv}}$ is a ring homomorphism, it also follows that $\psi_{\mathrm{cv}}^{-1}$ is a well-defined ring homomorphism $\mathrm{im}(\psi_{\mathrm{GF}}\circ\phi)\to \mathbb{Z}[s,t]/I_k$.

    Finally, we conclude that $\psi = \psi_{\mathrm{cv}}^{-1}\circ\psi_{\mathrm{GF}}\circ\phi$ is an isomorphism of ordered rings. This completes the proof of Theorem \ref{theorem:SU3k}.

    \subsection{SU(3)}

    The following result can be deduced from the work of Handelman--Rossmann 
    \cite{HR1,HR2} and Handelman \cite{H1,H2} (and is also related to the work 
    of Price \cite{P}).
More specifically, it is a corollary of the main result of \cite{H2}.

\begin{theorem}\label{theorem:SU3} As ordered rings,
    $$
        K_0(A(\mathbf{Rep}(\mathrm{SU}(3)),\pi_{(1,0)}))\cong \mathbb{Z}[s,t],
    $$
    where the positive cone on the right hand side is
    $$
        \mathbb{Z}[s,t]_+ = \{\sum_{a,b}s^at^bp_{a,b}(s,t)\,:\,p_{a,b}>0\text{ on }\mathcal{Y}\}\cup\{0\}
    $$
    for a certain compact set $\mathcal{Y}\subset [0,\infty)^{\times 2}$, which is described below, and the product on $K_0(A(\mathbf{Rep}(\mathrm{SU}(3)),\pi_{(1,0)}))\cong K_0(M_{3^\infty}^{\mathrm{SU}(3)})$ is induced by the $*$-homomor\-phism $M_{3^\infty}\otimes M_{3^\infty}\to M_{3^\infty}$ that interlaces the tensor factors.
    \end{theorem}

\begin{figure}
\centering
\resizebox{0.40\textwidth}{!}{%
\includegraphics{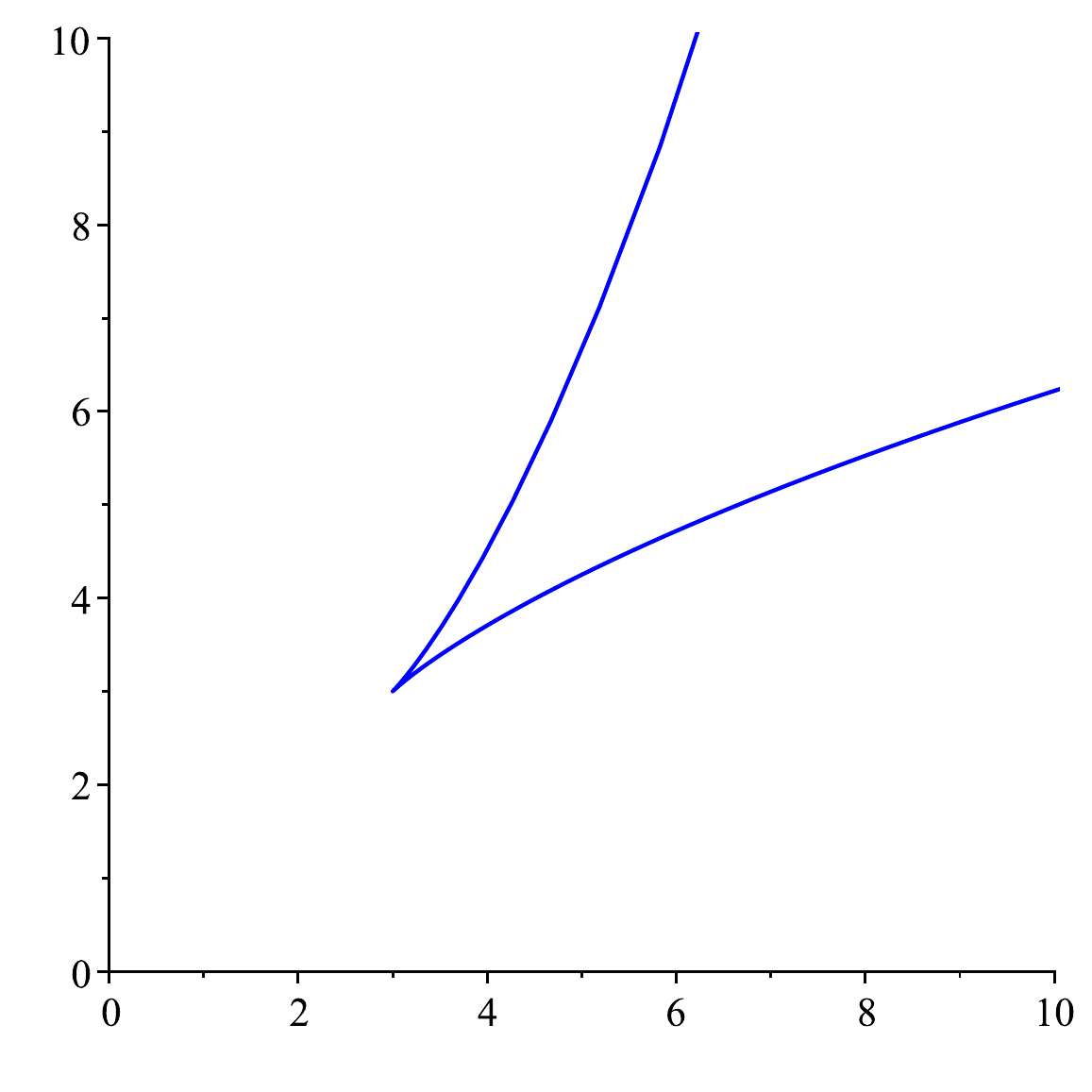}
}
\caption{\footnotesize The curve $(2a+a^{-2},2a^{-1}+a^2)$, $a>0$, which bounds $\mathcal{X}$. (The figure was produced using the computer program \emph{Maple}.)}
\label{fig:X}
\end{figure}

    \begin{remark}
    As in the case of $\mathbf{Rep}_k(\mathrm{SU}(3))$ (cf.\ Remark 
    \ref{remark:homom_A2}), the unital AF-algebra 
    $A(\mathbf{Rep}(\mathrm{SU}(3)),\pi_{(1,0)})$ is an inductive limit of 
    A$_2$-Temper\-ley--Lieb--Jones algebras and the $*$-homomorphism $\theta$ 
    has a description in terms of superimposed A$_2$-Temperley--Lieb diagrams.
    \end{remark}

    \noindent The set $\mathcal{Y}$ can be described as
    $$
    T(\mathcal{X})\cup \big(\{0\}\times [0,1/4]\big)\cup \big([0,1/4]\times\{0\}\big),
$$
where $\mathcal{X} = \{(a+b^{-1}+a^{-1}b,a^{-1}+b+ab^{-1})\,:\,a,b>0\}\subset(0,\infty)^{\times 2}$\label{page:X} and $T\colon (0,\infty)^{\times 2}\to (0,\infty)^{\times 2}$ is defined by $T(x,y) =
(x/y^2,y/x^2)$. The sets $\mathcal{X}$ and $\mathcal{Y}$ are depicted in figures
\ref{fig:X} and \ref{fig:Y}, respectively.

\begin{remark}\label{remark:SU3}
    Let us outline an ``elementary'' proof of Theorem \ref{theorem:SU3}. First of all, the asserted isomorphism $\psi$ is the unique map making the following diagram commute.
    $$
        \begin{CD}
            K_0(A(\mathbf{Rep}(\mathrm{SU}(3)),\pi_{(1,0)})) @>\psi >> \mathbb{Z}[s,t]\\
            @V\phi VV @V\psi_{\mathrm{cv}} VV\\
            F_{\mathbf{Rep}(\mathrm{SU}(3))}[\sigma^{-1}] @>\psi_{\mathrm{RR}}>> \mathbb{Z}[x^{\pm 1},y^{\pm 1}]
        \end{CD}
    $$
    Here, $\phi$ is the injective ring homomorphism from section \ref{subsection:ring_structure}, $\psi_{\mathrm{RR}}$ is induced by the classical ring isomorphism $F_{\mathbf{Rep}(\mathrm{SU}(3))}\to \mathbb{Z}[x,y]$ (cf.\ section \ref{subsection:fusion_rings}), and $\psi_{\mathrm{cv}}$ is defined by
    $
        \psi_{\mathrm{cv}}(p(s,t)) = p(x/y^2,y/x^2).
    $
    It is straightforward to verify that $\psi$ is a well-defined injective ring homomorphism. (Explicitly, $\psi$ is defined by a formula that is similar to that defining its namesake in section \ref{subsubsection:definition_psi}.) Since it is easy to see that $x/y^2$ and $y/x^2$ belong to $\mathrm{im}(\psi_{\mathrm{RR}}\circ \phi)$, hence that $1$, $s$ and $t$ belong to $\mathrm{im}\,\psi$, it follows that $\psi$ is surjective.

\begin{figure}
\centering
\resizebox{0.40\textwidth}{!}{%
\includegraphics{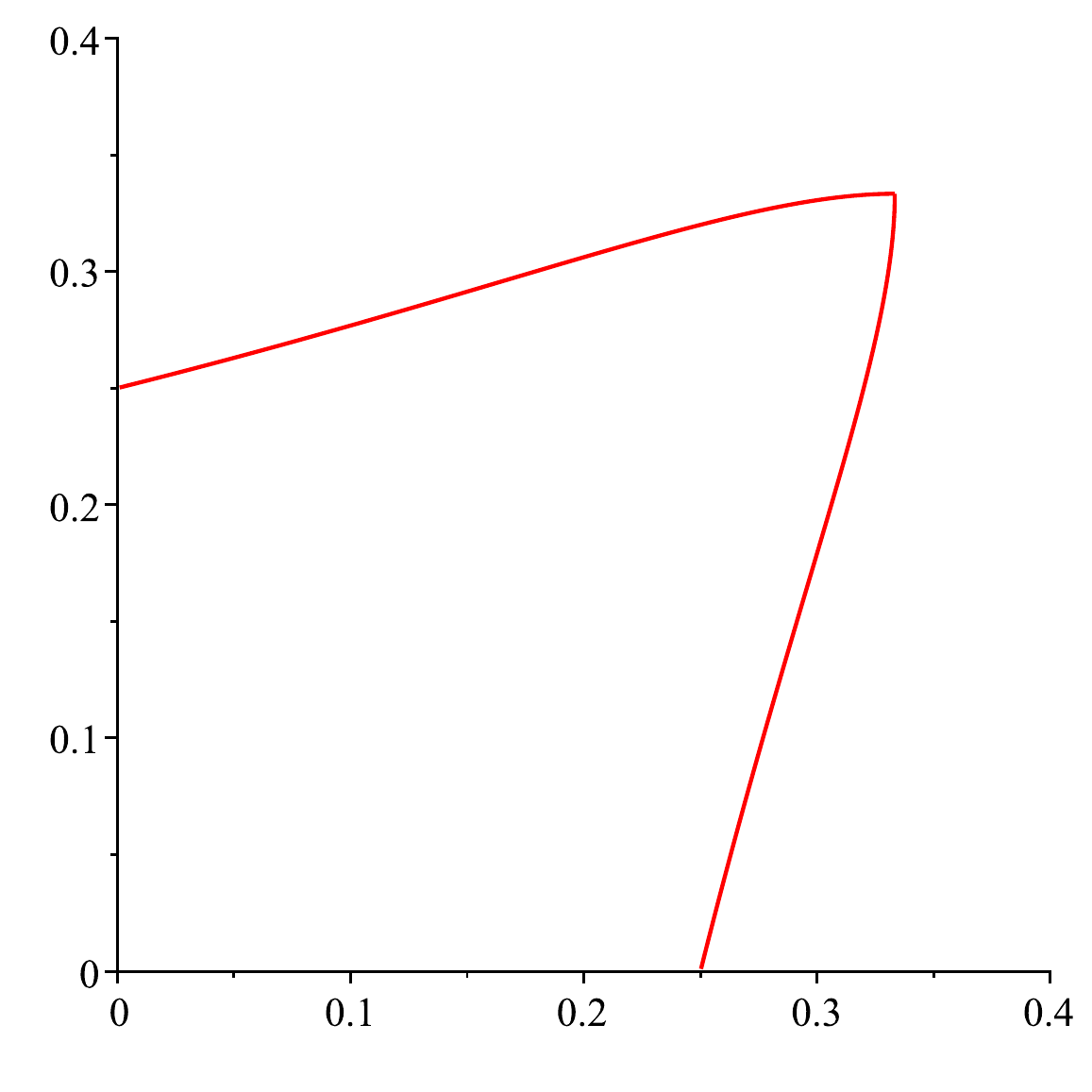}
}
\caption{\footnotesize The curve $\big((1-2u)/(2-3u)^2,u(2-3u)\big)$, $u\in
(0,1/3]$, and its
reflection in the line $s=t$, which together with the coordinate axes bound
$\mathcal{Y}$. (The figure was produced using the computer program
\emph{Maple}.)}
\label{fig:Y}
\end{figure}

    The non-trivial step in the proof is the determination of the positive cone $\mathbb{Z}[s,t]_+$, which can be accomplished by first using the well-known branching rules for the maximal torus in SU(3) as well as the Weyl Character Formula for SU(3), in the form of the formula
        \begin{multline*}\label{eq:2}
    Q^{\mathrm{SU}(3)}_{(n-1,m-1)}(z_1+z_2^{-1}+z_1^{-1}z_2,z_1^{-1}+z_2+z_1z_2^{-1})\\ = \frac{z_1^nz_2^m-z_1^{-m}z_2^{-n}+z_1^mz_2^{-(n+m)}-z_1^{n+m}z_2^{-m}+z_1^{-(n+m)}z_2^n-z_1^{-n}z_2^{n+m}}
    {z_1z_2-z_1^{-1}z_2^{-1}+z_1z_2^{-2}-z_1^{2}z_2^{-1}+z_1^{-2}z_2-z_1^{-1}z_2^2}
\end{multline*}
    (valid for all those $z_1,z_2\in\mathbb{C}\setminus\{0\}$ for which the denominator is non-zero and easily provable by induction), to show that
    \begin{align*}
        \mathcal{X}
         &= \{(x,y)\in\mathbb{R}^2\,:\,Q^{\mathrm{SU}(3)}_{\vec\lambda}(x,y)\geq 0\text{ for all }\vec\lambda\in\mathbb{N}_0^{\times 2}\}\\
         &= \{(x,y)\in\mathbb{R}^2\,:\,Q^{\mathrm{SU}(3)}_{\vec\lambda}(x,y) > 0\text{ for all }\vec\lambda\in\mathbb{N}_0^{\times 2}\}
    \end{align*}
    and then invoking two general facts, which we state below.

    First a bit of terminology. Let $(G,G_+)$ be an ordered group. We say that $(G,G_+)$ is \emph{unperforated} if $ng\in G_+$ implies $g\in G_+$ for any $g\in G$ and $n\in\mathbb{N}$. An \emph{order unit} $u$ in $(G,G_+)$ is an element $u\in G_+$ such that, for every $g\in G$, there exists $n\in\mathbb{N}$ for which $-nu\leq g\leq nu$. If $(G,G_+)$ is an ordered group with a distinguished order unit $u$ then a (normalized) \emph{state} on $(G,G_+)$ is a group homomorphism $G\to \mathbb{R}$ such that $\phi(g)\geq 0$ for all $g\in G_+$ and $\phi(u)=1$. The set of states is clearly convex, and it makes sense to speak of extremal states. By a state on an ordered ring, we shall mean a state on the underlying ordered group with (distinguished) order unit $1$.

    Now, we can state the promised general facts. (i) If an element of an 
    unperforated ordered group with an order unit (such as the ordered 
    $K_0$-group of any unital AF-algebra, cf.\ \cite{EHS}) has strictly 
    positive image under every extremal state then that element is positive (as 
    first proved by Effros--Handelman--Shen \cite{EHS} (item 1.4); see also 
    \cite{H1}, I.1). (ii) Every extremal state on an ordered ring is in fact a 
    ring homomorphism (as observed by Kerov--Vershik \cite{KV}, Maserick 
    \cite{M}, Voiculescu \cite{V}, and Wassermann \cite{Wa}; see also 
    \cite{H1}, I.2).
\end{remark}

As a corollary, we obtain the following result from \cite{Wa} (p.\ 123).

\begin{corollary}[Wassermann]\label{corollary:SO3} Denote by $\pi$ the defining representation of SO(3). Then, as ordered rings,
    $$
        K_0(A(\mathbf{Rep}(\mathrm{SO}(3)),\pi))\cong \mathbb{Z}[t],
    $$
    where the positive cone on the right hand side is
    $$
        \mathbb{Z}[t]_+ = \{p(t)\,:\,p>0\text{ on }(0,1/3]\}\cup\{0\}.
    $$
    Moreover, $A(\mathbf{Rep}(\mathrm{SO}(3)),\pi)\cong M_{3^\infty}^{\mathrm{SO}(3)}$ and, under the identification above and that in Theorem \ref{theorem:SU3}, the map induced by the inclusion $M_{3^\infty}^{\mathrm{SU}(3)}\to M_{3^\infty}^{\mathrm{SO}(3)}$ in $K$-theory is the ring homomorphism $\mathbb{Z}[s,t]\to\mathbb{Z}[t]$ given by $p(s,t)\mapsto p(t,t)$.
\end{corollary}

\begin{proof}
    This follows easily from Theorem \ref{theorem:SU3} and the well-known branching rules for the inclusion $\mathrm{SO}(3)\subset\mathrm{SU}(3)$ (cf.\ \cite{Ra}).
\end{proof}

    \subsection{Sp(4)$_k$}\label{subsection:Sp4k}

    The same techniques as those employed for SU(3)$_k$ yield the following result. As the proof introduces no new ideas, we omit it from the present paper. We note, however, that it is based on the formula for the fusion ideals $J_k(\mathrm{Sp}(4))$ that we mentioned in section \ref{subsection:fusion_rings} (and, of course, on the well-known fusion graph of $\mathbf{Rep}_k(\mathrm{Sp}(4))$ with respect to $\pi_{(0,1)}$, which we also omit).

    \begin{theorem}\label{theorem:Sp4k}
    As ordered rings,
    $$
        K_0(A(\mathbf{Rep}_k(\mathrm{Sp}(4)),\pi_{(0,1)}))\cong \mathbb{Z}[s,t]/I_k,
    $$
    where $I_k = \langle \{P_{\vec\lambda}^{\mathrm{Sp}(4)}(s,t)\,:\,\lambda_1+\lambda_2 = k+1\}\cup\{P_{(0,k+2)}^{\mathrm{Sp}(4)}(s,t)\}\rangle$.

    Here, the polynomial $P_{\vec\lambda}^{\mathrm{Sp}(4)}(s,t)$ (for $\vec\lambda\in \mathbb{N}_0^{\times 2}$) is obtained from the Laurent polynomial $y^{-(\lambda_1+\lambda_2)}Q_{\vec\lambda}^{\mathrm{Sp}(4)}(x,y)$ by performing the change of variables $(s,t) = (1/y,x^2/y^{2})$, and the positive cone on the right hand side is
    $$
        \big(\mathbb{Z}[s,t]/I_k\big)_+ = \{[p(s,t)]\,:\,p(\beta_{k,1},\beta_{k,2})>0\}\cup\{[0]\},
    $$
    where $\beta_{k,1} = d(\pi_{(0,1)})^{-1}$ and $\beta_{k,2} = d(\pi_{(1,0)})^2d(\pi_{(0,1)})^{-2}$.
    \end{theorem}

    \noindent The positive cone is well-defined by the discussion in section \ref{subsection:fusion_rings}. Explicit formulae for $\beta_{k,1}$ and $\beta_{k,2}$, which may e.g.\ be deduced from the formulae for the modular $S$-matrix that are mentioned in section \ref{subsection:fusion_rings}, are rather complicated and not very illuminating, so we omit them.

    \subsection{(G$_{2}$)$_k$}

    Similarly, we were able to prove the following result.

    \begin{theorem}\label{theorem:G2k}
    As ordered rings,
    $$
        K_0(A(\mathbf{Rep}_k(\mathrm{G}_2),\pi_{(1,0)}))\cong \mathbb{Z}[s,t]/I_k,
    $$
    where $I_k = \langle \{P_{\vec\lambda}^{\mathrm{G}_2}(s,t)\,:\,\lambda_1+2\lambda_2 = k+1\}\cup\{P_{(k+2,0)}^{\mathrm{G}_2}(s,t)\}\rangle$.

    Here, the polynomial $P_{\vec\lambda}^{\mathrm{G}_2}(s,t)$ (for $\vec\lambda\in \mathbb{N}_0^{\times 2}$) is obtained from the Laurent polynomial $x^{-(\lambda_1+2\lambda_2)}Q_{\vec\lambda}^{\mathrm{G}_2}(x,y)$ by performing the change of variables $(s,t) = (1/x,y/x^{2})$, and the positive cone on the right hand side is
    $$
        \big(\mathbb{Z}[s,t]/I_k\big)_+ = \{[p(s,t)]\,:\,p(\beta_{k,1},\beta_{k,2})>0\}\cup\{[0]\},
    $$
    where $\beta_{k,1} = d(\pi_{(1,0)})^{-1}$ and $\beta_{k,2} = d(\pi_{(1,0)})^{-2}d(\pi_{(0,1)})$.
    \end{theorem}

    \noindent Similar comments to those in section \ref{subsection:Sp4k} apply.

    \section{Recovering the Verlinde ring in special cases}\label{section:recovering}

    We will in this section show that $K_0(A(\mathcal{C},\pi))\cong F_{\mathcal{C}}$ as ordered rings in certain special cases, where $F_{\mathcal{C}}$ is equipped with the positive cone $(F_{\mathcal{C}})_+ = \{x\,:\,d(x)>0\}\cup\{0\}$. (As usual, $d$ denotes the quantum dimension in $\mathcal{C}$.)

    \subsection{General remarks}

    Fix a rigid C*-tensor category $\mathcal{C}$ and an object $\pi$ therein. We denote by $\Lambda$ the set of simple objects in $\mathcal{C}$ and (as in Remark \ref{remark:subcategory}) by $\mathcal{C}_1$ the full C*-tensor subcategory of $\mathcal{C}$ generated by the simple objects that occur as direct summands in tensor powers of $\sigma = \bar{\pi}\otimes\pi$.
    In the next two lemmas, we clarify some assumptions that we will make later.

    \begin{lemma}\label{lemma:complete}
        Suppose that $\Lambda$ is a finite set. Then the following conditions are equivalent.
        \begin{itemize}
            \item[(i)] $\mathcal{C}_1 = \mathcal{C}$;
            \item[(ii)] For every $\mu\in\Lambda$ there exists $n\in\mathbb{N}_0$ such that $\mu\prec\sigma^{\otimes n}$;
            \item[(iii)] There exists $n\in\mathbb{N}_0$ such that $\mu\prec\sigma^{\otimes n}$ for every $\mu\in\Lambda$.
        \end{itemize}
    \end{lemma}

    \begin{proof}
        This follows from the fact that $1\prec\sigma$, which implies that the set of those $\mu\in\Lambda$ for which $\mu\prec\sigma^{\otimes n}$ increases with $n$, and hence stabilizes eventually.
    \end{proof}

    \noindent Given a commutative ring $R$ with identity, we will use the symbol $\mathrm{Inv}(R)$ to denote the group of invertible elements in $R$.

    \begin{lemma}\label{lemma:invertible}
        Suppose that $\Lambda$ is a finite set and that $F_{\mathcal{C}}$ is commutative. Then the following conditions are equivalent.
        \begin{itemize}
            \item[(i)] $\det(N_\pi) = \pm 1$;
            \item[(ii)] $\sigma\in\mathrm{Inv}(F_{\mathcal{C}})$;
            \item[(iii)] $\sigma\in\mathrm{Inv}(F_{\mathcal{C}_1})$.
        \end{itemize}
    \end{lemma}

    \begin{proof}
        Note that $F_{\mathcal{C}}\otimes\mathbb{C}$ is a finite-dimensional $*$-algebra under the $*$-operation induced by conjugation of objects.
        Recall the regular representation of $\mathcal{C}$ (in the terminology of \cite{PV}), which is the injective $*$-homomorphism $N\colon F_{\mathcal{C}}\otimes\mathbb{C}\to M_\Lambda(\mathbb{C})$ defined by $N(\mu) = N_\mu$ for $\mu\in\Lambda$ (and extended by linearity), which manifests $F_{\mathcal{C}}\otimes\mathbb{C}$ as a commutative C*-subalgebra of $M_{\Lambda}(\mathbb{C})$.

        Let $x\in F_{\mathcal{C}}$ be given. Then $N(x)\in M_{\Lambda}(\mathbb{Z})$ and we claim that $\det(N(x))=\pm 1$ if and only if $x\in\mathrm{Inv}(F_{\mathcal{C}})$.  To see this, assume first that $\det(N(x))=\pm 1$. Then $N(x)$ is invertible in $M_{\Lambda}(\mathbb{Z})$, whereby it is also invertible in the C*-subalgebra $F_{\mathcal{C}}\otimes\mathbb{C}$ of $M_{\Lambda}(\mathbb{C})$. Thus, there are complex numbers $c_\mu$ such that
    $
        x(\sum_{\mu}c_\mu\mu) = 1.
    $
    Applying $N$ to this equation, we get that $N(x)\big(\sum_{\mu} c_\mu N_\mu\big) = 1$, i.e., $N(x)^{-1} = \sum_{\mu} c_\mu N_\mu$. In particular, the columns corresponding to $\boldsymbol{1}\in\Lambda$ coincide. It follows that $c_\mu\in\mathbb{Z}$ for all $\mu\in\Lambda$, hence that $x\in \mathrm{Inv}(F_{\mathcal{C}})$. The converse implication is clear, so the claim is now proved.

The claim implies that (i) is equivalent to (ii). Moreover, it follows from the claim and Frobenius Reciprocity that (ii) is equivalent to (iii).
    \end{proof}

    \begin{remark}
        Note that, in the case where $\mathcal{C} = \mathbf{Rep}_k(G)$, one can 
        define a\linebreak $*$-isomorphism 
        $\mathrm{Ver}_k(G)\otimes\mathbb{C}\to C(V)$ by $\pi_{\vec\lambda} 
        \mapsto Q_{\vec\lambda}^G(x_1,\ldots,x_r)|_{V}$, where $V$ is the 
        fusion variety associated to $\mathbf{Rep}_k(G)$, cf.\ section 
        \ref{subsection:fusion_rings}. Moreover, in the notation of that 
        section, $Q^G_{\vec\lambda}(x_{\vec\nu}^{(1)},\ldots,x_{\vec\nu}^{(r)}) 
        = S_{\vec\lambda,\vec\nu}/S_{\vec 0,\vec\nu}$, which means that the 
        vectors $(S_{\vec\lambda,\vec\nu})_{\vec\lambda}$, i.e., the columns of 
        the modular $S$-matrix, constitute a basis of simultaneous eigenvectors 
        for the fusion matrices $N_{\pi_{\vec\lambda}}$ and, moreover, that the 
        eigenvalue of $N_{\pi_{\vec\lambda}}$ that corresponds to the 
        eigenvector $(S_{\vec\lambda,\vec\nu})_{\vec\lambda}$ is 
        $S_{\vec\lambda,\vec\nu}/S_{\vec 0,\vec\nu}$. This is precisely the 
        statement that the Verlinde Conjecture (cf.\ \cite{Ver}) holds for the 
        WZW models, a fact which was proved in \cite{TUY}, \cite{Fa}, \cite{Be}.
    \end{remark}

    \begin{proposition}\label{proposition:fusion_ring}
        Suppose that $\mathcal{C}$ is a braided C*-tensor category with finitely many simple objects and suppose that the given object $\pi$ in $\mathcal{C}$ satisfies both the equivalent conditions in Lemma \ref{lemma:complete} and those in Lemma \ref{lemma:invertible}. Then
        $$
            K_0(A(\mathcal{C},\pi))\cong F_{\mathcal{C}}
        $$
        as ordered rings, where the positive cone in $F_{\mathcal{C}}$ is $(F_{\mathcal{C}})_+ = \{x\,:\,d(x)>0\}\cup\{0\}$.
    \end{proposition}

    \begin{proof}
        This follows immediately from Remark \ref{remark:ordered_ring_iso}.
    \end{proof}

    \subsection{SU(2)$_k$}

    We will next give some examples of ``levels'' $k$ and (isomorphism classes of) objects $\pi$ in the category $\mathcal{C} = \mathbf{Rep}_k(\mathrm{SU}(2))$ for which the hypotheses of Proposition \ref{proposition:fusion_ring} are satisfied.

    \begin{lemma} The element $\pi = 1 + \pi_1\in\mathrm{Ver}_k(\mathrm{SU}(2))$ satisfies
    $$
        \pi\in\mathrm{Inv}(\mathrm{Ver}_k(\mathrm{SU}(2)))\iff k\notin 1+3\mathbb{Z}.
    $$
    \end{lemma}

    \begin{proof} Put $a_k = \det(N_{\pi})$, where $\pi$ is viewed as an object in $\mathbf{Rep}_k(\mathrm{SU}(2))$. It is easy to verify that $a_1 = 0$, $a_2 = -1$, and $a_{k+1} = a_k-a_{k-1}$ for all $k\geq 2$. It follows that $a_k = 0$ if $k\in 1+3\mathbb{Z}$ while $a_k = \pm 1$ otherwise.
    \end{proof}

    \noindent Since every simple object in $\mathbf{Rep}_k(\mathrm{SU}(2))$ 
    occurs as a direct summand of\linebreak 
    $(\boldsymbol{1}\oplus\pi_1)^{\otimes k}$, the object 
    $\boldsymbol{1}\oplus\pi_1$ satisfies the hypotheses of Proposition 
    \ref{proposition:fusion_ring}. Thus,
    if $k\notin 1+3\mathbb{Z}$ then $K_0(A(\mathbf{Rep}_k(\mathrm{SU}(2)),\boldsymbol{1}\oplus\pi_1))\cong \mathrm{Ver}_k(\mathrm{SU}(2))$ as ordered rings.

    \begin{lemma}\label{lemma:even_identity} In the ring $\mathrm{Ver}_{2n}(\mathrm{SU}(2))$, where $n\in\mathbb{N}_0$, we have
    that
    $$
    	\sum_{j=0}^{2n}(-1)^j\pi_j^2 = \sum_{j=0}^{n}(-1)^j\pi_{2j}.
    $$
    \end{lemma}

    \begin{proof} Assume first that $2n = 4m-2$. Then we have that
    $$
    	\sum_{j=0}^{4m-2}(-1)^j\pi_j^2 =
    	\sum_{j=0}^{2m-1}(-1)^j[1+\pi_2+\cdots+\pi_{2j}]+\sum_{j=2m}^{4m-2}(-1)^j[1+\pi_2+\cdots+\pi_{8m-4-2j}].
    $$
    In the first sum on the right hand side, $1$ appears $2m$ times with
    alternating sign, hence cancels out. More generally, for each $0\leq i\leq
    2m-1$, $\pi_{2i}$ appears
    $2m-i$ times with alternating sign, starting with the sign $(-1)^i$. Thus,
    the first sum is equal to $-\pi_2-\pi_6-\cdots-\pi_{4m-2}$. Similarly, the
    second sum on the right hand side is equal to
    $1+\pi_4+\cdots+\pi_{4m-4}$. This proves the formula in this case. The
    proof in the case where $2n=4m$ is similar.\end{proof}

    \begin{proposition} In the ring $\mathrm{Ver}_{2n}(\mathrm{SU}(2))$, where $n\in\mathbb{N}_0$, we have
        the identity
        $$
        	\left(1+2\sum_{j=1}^{2n}\pi_j\right)\left(1+2\sum_{j=1}^{2n}(-1)^j\pi_j\right)
        	 = 1.
        $$
    \end{proposition}

    \begin{proof} It is clearly enough to show that
        $$
            \sum_{i,j=1}^{2n}(-1)^i\pi_i\pi_j = -\sum_{j=1}^{n}\pi_{2j}.
        $$
        By Lemma \ref{lemma:even_identity}, this is equivalent to
        $$
            \sum_{\substack{1\leq i<j\leq 2n\\ i\equiv j\text{ (mod 2)}}}
            (-1)^i\pi_i\pi_j = -\sum_{i=1}^{\lfloor n/2\rfloor} \pi_{4i}.
        $$
        We can write the left hand side as $\sum_{r=1}^{n-1}x_r$, where
        $$
        	x_r = \sum_{i=1}^{n-r}(-1)^i\pi_i\pi_{i+2r}+\sum_{i=n-r+1}^{2n-2r}(-1)^i\pi_i\pi_{i+2r}.
        $$
        Here, the left-most sum is equal to
        $$
        	\sum_{i=1}^{n-r}(-1)^i[\pi_{2r}+\pi_{2r+2}+\cdots+\pi_{2r+2i}]
        $$
        while the right-most sum is equal to
        $$
        	\sum_{i=n-r+1}^{2n-2r}(-1)^i[\pi_{2r}+\pi_{2r+2}+\cdots+\pi_{4n-2r-2i}].
        $$
        Thus, by the proof of Lemma \ref{lemma:even_identity}, if $r\equiv n+ 1$ (mod 2) then
        $
        	x_r = -\pi_{2r+2}+\pi_{2r+4}-\cdots-\pi_{2n}
        $
        while if $r\equiv n$ (mod 2) then
        $
        	x_r = -\pi_{2r+2}+\pi_{2r+4}-\cdots+\pi_{2n}.
        $
        It follows easily from this that $\sum_{r=1}^{n-1}x_r = -\sum_{i=1}^{\lfloor n/2\rfloor} \pi_{4i}$, as desired.
    \end{proof}

    \noindent We conclude that the object $\pi = \boldsymbol{1}\oplus \bigoplus_{j=1}^{2n}\pi_j^{\oplus 2}$ in $\mathbf{Rep}_{2n}(\mathrm{SU}(2))$ satisfies the hypotheses of Proposition \ref{proposition:fusion_ring}.

    \begin{remark}
        If  $n\geq 2$ then the ordered group 
        $\mathrm{Ver}_{1}(\mathrm{SU}(n))$, equipped with the positive cone 
        defined above, is not the ordered $K_0$-group of any AF-algebra. This 
        can be shown as follows. Recall first that an ordered group $(G,G_+)$ 
        is said to have the \emph{Riesz interpolation property} if, whenever 
        $h_1,h_2,g_1,g_2\in G$ are such that $h_i\leq g_j$ for all 
        $i,j\in\{1,2\}$, there exists $x\in G$ such that $h_i\leq x\leq g_j$ 
        for all $i,j\in\{1,2\}$. It is a fact that the ordered $K_0$-group of 
        any AF-algebra has this property, cf.\ \cite{EHS}. To see that 
        $\mathrm{Ver}_{1}(\mathrm{SU}(n))$ with the positive cone 
        $\mathrm{Ver}_{1}(\mathrm{SU}(n))_+ = \{x\,:\,d(x)>0\}\cup\{0\}$ does 
        not have this property, note first that each of the $n$ simple objects 
        $\rho_1 = \boldsymbol{1}$, $\rho_2 = \pi_{\vec{e}_1}$, $\ldots$, 
        $\rho_n = 
        \pi_{\vec{e}_{n-1}}$ in $\mathbf{Rep}_1(\mathrm{SU}(n))$ has quantum 
        dimension $1$ by the main result of \cite{F1}. Moreover, with $h_1 = 
        0$, $h_2 = \rho_1-\rho_2$, $g_1 = \rho_1$ and $g_2 = 2\rho_1-\rho_2$, 
        we get that $h_i\leq g_j$ for all $i,j\in\{1,2\}$. If now $x\in G$ were 
        such that $h_i\leq x\leq g_j$ for all $i,j\in\{1,2\}$ then we would 
        either have $d(x)=0$ or $d(x)=1$, both of which lead to a contradiction.
    \end{remark}

    \section{Concluding remarks}\label{section:conclusion}

    In the present paper, we have explicitly identified certain ordered $K_0$-groups with (quotients of) polynomial rings over $\mathbb{Z}$ and shown that the product on these rings is induced by a $*$-homomorphism that arises from a unitary braiding on an underlying rigid C*-tensor category. These considerations raise the following question. As $A(\mathbf{Rep}(G),\pi)\cong M_{n^\infty}^G$, where $n = \dim(\pi)$, one is lead to ask if also $A(\mathbf{Rep}_k(G),\pi)$ may be realized as a fixed point algebra $A^{G_q}$ under some action (properly defined) of a quantized deformation $G_q$ of $G$ on a C*-algebra $A$. If we could do this, then the ring structure on $K_0(A^{G_q})$ should be induced by a ``natural'' $*$-homomorphism $A\otimes A\to A$.

    Wassermann \cite{Wa} and Handelman--Rossmann \cite{HR2} observed that if 
    $G$ is a compact, connected group then $K_0(M_{n^\infty}\rtimes G)$ may be 
    identified with the localization $R(G)[(\bar{\pi}\pi)^{-1}]$, where $R(G) = 
    F_{\mathbf{Rep}(G)}$ is the representation ring of $G$. Moreover, they 
    noted that $K_0(M_{n^\infty}^G)$ may be identified with a certain subring 
    of $K_0(M_{n^\infty}\rtimes G)$ (see also \cite{R}). In analogy with this, 
    and with the isomorphism $K_0^G(M_{n^\infty})\cong K_0(M_{n^\infty}\rtimes 
    G)$, which follows from the Green--Julg Theorem (cf.\ e.g.\ \cite{Ph}), we 
    ought to have that
    $$
        K_0^{G_q}(A)\cong K_0(A\rtimes G_q)\cong \mathrm{Ver}_k(G)[(\bar{\pi}\pi)^{-1}].
    $$
    Based on the compact group case, one can guess what a Bratteli diagram for $A\rtimes G_q$ should look like. Namely, every level should have the vertex set $\Lambda$, i.e., the set of simple objects in $\mathbf{Rep}_k(G)$, and the inclusion matrix between every pair of neighboring levels should be the fusion matrix $N_{\bar{\pi}\otimes\pi} = N_\pi^*N_\pi$. Denote by $B(\mathbf{Rep}_k(G),\pi)$ the AF-algebra arising from this Bratteli diagram. (We are on purpose being vague about the multiplicities of the vertices.) It is then trivial to show that $K_0(B(\mathbf{Rep}_k(G),\pi))\cong \mathrm{Ver}_k(G)[\sigma^{-1}]$ as ordered groups. It remains to figure out how to identify $B(\mathbf{Rep}_k(G),\pi)$ with $A\rtimes G_q$ once the action of $G_q$ on $A$ has been defined and how to equip $K_0^{G_q}(A)$ with a ``natural'' ring structure.

    Let us describe the positive cone of the ordered group $K_0(B(\mathbf{Rep}_k(G),\pi))$ in some special cases. The work of Evans and Gould (cf.\ Theorem \ref{theorem:SU2k}) readily implies that positivity in $K_0(B(\mathbf{Rep}_k(\mathrm{SU}(2)),\pi_1))\cong \mathrm{Ver}_k(\mathrm{SU}(2))[\pi_1^{-2}]\cong \mathbb{Z}[x^{\pm 1}]/\tilde{J}_k$ is determined by evaluating the even and odd parts of a coset separately at $d(\pi_1)$.
    (Here, $\tilde{J}_k$ is the image of the fusion ideal $J_k(\mathrm{SU}(2))$ in the localization $\mathbb{Z}[x^{\pm 1}]$ of $\mathbb{Z}[x]$. In what follows, $\tilde{J}_k$ will always denote the appropriate localization of the relevant fusion ideal.) More precisely, a coset $[f(x)]$ in $\mathbb{Z}[x^{\pm 1}]/\tilde{J}_k$ belongs to the positive cone if and only if $[f(x)] = [f_0(x)]+[f_1(x)]$, where $f_0(x)$ is even, $f_1(x)$ is odd and, for both $i=1,2$, $[f_i(x)] = [0]$ or $f_i(d(\pi_1))>0$. Note that this splitting up into two parts is due to the fact that the center $\mathcal{Z}(\mathrm{SU}(2))$ of $\mathrm{SU}(2)$ is isomorphic to $\mathbb{Z}/2\mathbb{Z}$.

    Similarly, positivity in $K_0(B(\mathbf{Rep}_k(\mathrm{SU}(3)),\pi_{(1,0)}))\cong \mathbb{Z}[x^{\pm 1},y^{\pm 1}]/\tilde{J}_k$ is determined by evaluating the parts of type $0$, $1$ and $2$ separately at $(d(\pi_{(1,0)}),d(\pi_{(1,0)}))$. Here, a monomial $x^ay^b$ with $a,b\in\mathbb{Z}$ is said to be of type $j$ if $a-b\equiv j$ (mod 3), where the `3' is related to the fact that $\mathcal{Z}(\mathrm{SU}(3))\cong\mathbb{Z}/3\mathbb{Z}$. Finally, positivity in $K_0(B(\mathbf{Rep}_k(\mathrm{Sp}(4)),\pi_{(0,1)}))\cong \mathbb{Z}[x,y^{\pm 1}]/\tilde{J}_k$ is determined by evaluation of ``semi-even'' and ``semi-odd'' parts (due to $\mathcal{Z}(\mathrm{Sp}(4))\cong \mathbb{Z}/2\mathbb{Z}$), where $x^ay^b$ is semi-even (resp.\ semi-odd) if $a$ is even (resp.\ odd), while positivity in $K_0(B(\mathbf{Rep}_k(\mathrm{G}_2),\pi_{(1,0)}))\cong \mathbb{Z}[x^{\pm 1},y]/\tilde{J}_k$ is determined by evaluation of the whole function (as $\mathrm{G}_2$ has trivial center).

    These positivity conditions are analogous to those for the crossed products of the compact groups in question. For instance, it is easily deduced from the work of Wassermann (cf.\ Theorem \ref{theorem:SU2} above) that positivity in $K_0(M_{2^\infty}\rtimes\mathrm{SU}(2))\cong R(\mathrm{SU}(2))[\pi_1^{-1}]\cong \mathbb{Z}[x^{\pm 1}]$ is determined by evaluating the even and odd parts of a function separately on the interval $[2,\infty)$. Moreover, positivity in $K_0(M_{3^\infty}\rtimes\mathrm{SU}(3))\cong \mathbb{Z}[x^{\pm 1},y^{\pm 1}]$ is more or less (but not quite) determined by evaluating the parts of type $0$, $1$ and $2$ separately on the set $\mathcal{X}$ that was defined on page \pageref{page:X}.

    Let us next discuss some exact sequences that relate $K_0(B(\mathbf{Rep}_k(G),\pi))$ to $\mathrm{Ver}_k(G)$. Given $x\in\mathrm{Ver}_k(G)\otimes\mathbb{Q}$, there is a short exact sequence of rings
    $$
        0\to \mathrm{Ann}_{\mathrm{Ver}_k(G)\otimes\mathbb{Q}}(x)\to \mathrm{Ver}_k(G)\otimes\mathbb{Q}\to (\mathrm{Ver}_k(G)\otimes\mathbb{Q})[x^{-1}]\to 0,
    $$
    where $\mathrm{Ann}_{\mathrm{Ver}_k(G)\otimes\mathbb{Q}}(x) = \{y\in\mathrm{Ver}_k(G)\otimes\mathbb{Q}\,:\,xy = 0\}$ is isomorphic to the null-space of $N(x)$ (cf.\ the proof of Lemma \ref{lemma:invertible}) as a vector space over $\mathbb{Q}$.
        In the case where $G = \mathrm{SU}(2)$, $x = \pi_1$ and $k$ is even, this yields a short exact sequence
    \begin{equation}\label{equation:SES_SU2}
        0\to \mathbb{Q}\to \mathrm{Ver}_k(\mathrm{SU}(2))\otimes\mathbb{Q}\to K_0(B(\mathbf{Rep}_k(\mathrm{SU}(2)),\pi_1))\otimes\mathbb{Q}\to 0
    \end{equation}
    of groups, since $$(\mathrm{Ver}_k(\mathrm{SU}(2))\otimes\mathbb{Q})[\pi_1^{-1}]\cong \mathbb{Q}[x^{\pm 1}]/\tilde{J}_k\cong K_0(B(\mathbf{Rep}_k(\mathrm{SU}(2)),\pi_1))\otimes\mathbb{Q}$$ and $N_{\pi_1}$ has nullity one. (Here, $\tilde{J}_k$ is the ideal obtained by taking the image of the fusion ideal $J_k(\mathrm{SU}(2))$ in the localization $\mathbb{Q}[x^{\pm 1}]$ of $\mathbb{Z}[x]$.) However, when $k$ is odd, $N_{\pi_1}$ is invertible and we get that $K_0(B(\mathbf{Rep}_k(\mathrm{SU}(2)),\pi_1))\otimes\mathbb{Q}\cong \mathrm{Ver}_k(\mathrm{SU}(2))\otimes\mathbb{Q}$. In fact, one can show that $$K_0(B(\mathbf{Rep}_k(\mathrm{SU}(2)),\pi_1))\cong \mathrm{Ver}_k(\mathrm{SU}(2))$$ in this case, since $N_{\pi_1}$ is actually invertible over $\mathbb{Z}$, and use the identification $$K_0(B(\mathbf{Rep}_k(\mathrm{SU}(2)),\pi_1))\cong K_0(A(\mathbf{Rep}_k(\mathrm{SU}(2)),\pi_1))^{\oplus 2}$$ of groups to define ``intrinsically'' a ring structure on $K_0(B(\mathbf{Rep}_k(\mathrm{SU}(2)),\pi_1))$ giving rise to that of $\mathrm{Ver}_k(\mathrm{SU}(2))$.

    Similarly, in the case where $G = \mathrm{SU}(3)$, $x = \pi_{(1,0)}$ and $k\in 3\mathbb{Z}$, we get a short exact sequence
    \begin{equation}\label{equation:SES_SU3}
        0\to \mathbb{Q}\to \mathrm{Ver}_k(\mathrm{SU}(3))\otimes\mathbb{Q}\to K_0(B(\mathbf{Rep}_k(\mathrm{SU}(3)),\pi_1))\otimes\mathbb{Q}\to 0
    \end{equation}
    of groups. If $k\notin 3\mathbb{Z}$ then $N_{\pi_{(1,0)}}$ is invertible over $\mathbb{Q}$ and so $$K_0(B(\mathbf{Rep}_k(\mathrm{SU}(3)),\pi_{(1,0)}))\otimes\mathbb{Q}\cong \mathrm{Ver}_k(\mathrm{SU}(3))\otimes\mathbb{Q}.$$ (In this case, it seems that $K_0(B(\mathbf{Rep}_k(\mathrm{SU}(3)),\pi_{(1,0)}))\not\cong \mathrm{Ver}_k(\mathrm{SU}(3))$. At least $[x]$ is not an invertible element of $\mathbb{Z}[x,y]/J_k(\mathrm{SU}(3))\cong \mathrm{Ver}_k(\mathrm{SU}(3))$ unlike in $\mathbb{Z}[x^{\pm 1},y^{\pm 1}]/\tilde{J}_k\cong K_0(B(\mathbf{Rep}_k(\mathrm{SU}(3)),\pi_{(1,0)}))$.)

    For $G = \mathrm{Sp}(4)$ and $G = \mathrm{G}_2$, the short exact sequences we get are less nice, since the nullity of the fusion matrices of fundamental representations seems to either fluctuate in a haphazard manner or to never be zero as indicated by some experimentation using the computer program \emph{Maple}. However, we do get that
    $$
        K_0(B(\mathbf{Rep}_k(\mathrm{Sp}(4)),\pi_{(0,1)}))\otimes\mathbb{Q}\cong \mathrm{Ver}_k(\mathrm{Sp}(4))\otimes\mathbb{Q}
    $$
    whenever $N_{\pi_{(0,1)}}$ is invertible over $\mathbb{Q}$, which seems likely to occur if and only if $k+3$ is divisible by neither $3$ nor $5$. Indeed, we used \emph{Maple} to verify this for $k\leq 100$. Similarly, we used \emph{Maple} to verify, again for all $k\leq 100$, that $N_{\pi_{(1,0)}}$ is not invertible. More precisely, we verified that the nullity of $N_{\pi_{(1,0)}}$ is $\lfloor(k+2)/2\rfloor$ for all $k\leq 100$. We also get that
    $$
        K_0(B(\mathbf{Rep}_k(\mathrm{G}_2),\pi_{(1,0)}))\otimes\mathbb{Q}\cong \mathrm{Ver}_k(\mathrm{G}_2)\otimes\mathbb{Q}
    $$
    whenever $N_{\pi_{(1,0)}}$ is invertible over $\mathbb{Q}$. A theorem of Gannon and Walton (cf.\ Theorem 7 in \cite{GW}) states that this is the case precisely when $k+4$ is divisible by neither $4$ nor $30$. However, we used \emph{Maple} to verify that, for $k\leq 100$, $N_{\pi_{(1,0)}}$ is invertible if and only if $k+4$ is divisible by neither $4$, $7$ nor $30$. Moreover, we verified by hand that $N_{\pi_{(1,0)}}$ is not invertible over $\mathbb{Q}$ when $k = 3$. Similarly, a computation using \emph{Maple} showed that, for $k\leq 100$, $N_{\pi_{(0,1)}}$ is invertible if and only if $k+4$ is divisible by neither $5$, $7$ nor $8$.

    Now, the short exact sequences in equations (\ref{equation:SES_SU2}) and (\ref{equation:SES_SU3}) as well as the various isomorphisms above lead us to ask if they can somehow be understood on the level of $*$-homomorphisms between (possibly non-AF) C*-algebras. This is related to our aforementioned long-term goal of finding a ``natural'' C*-algebra $B$ such that $K_0(B)\cong \mathrm{Ver}_k(G)$ as rings and $\mathbf{Rep}_k(G)$ is the category of (finitely generated projective) modules over $B$ under an appropriate tensor product (see also \cite{AE}).

\vspace{2mm}
{\footnotesize
\noindent {\sc Andreas N\ae s Aaserud}\\ %({\sl ORCID: 0000-0001-9685-1640}\,)\\
E-mail: aaseruda@cardiff.ac.uk / andreas.naes.aaserud@gmail.com\\[1mm]
{\sc David E.\ Evans}\\
E-mail: evansde@cardiff.ac.uk\\[1.5mm]
\noindent {\it Both authors are affiliated with}\\[1mm]
{\sc School of Mathematics, Cardiff
University, Senghennydd Road, Cardiff, CF24 4AG, Wales, UK}
}
\end{document}